\documentclass[12pt,a4paper]{article}
\usepackage{amsmath,amsthm,amsfonts,amssymb,bbm}
\usepackage{graphicx,psfrag,subfigure,color}
\usepackage{cite}
\usepackage{hyperref}

\numberwithin{equation}{section}

\DeclareFontFamily{OT1}{pzc}{}
\DeclareFontShape{OT1}{pzc}{m}{it}{<-> s * [1.10] pzcmi7t}{}
\DeclareMathAlphabet{\mathpzc}{OT1}{pzc}{m}{it}

\newcommand{\Or}{\mathcal{O}}
\newcommand{\Ai}{\mathrm{Ai}}
\newcommand{\He}{\mathrm{H}}
\newcommand{\const}{\mathrm{const}}
\newcommand{\Pb}{\mathbbm{P}}
\newcommand{\E}{\mathbbm{E}}
\newcommand{\Id}{\mathbbm{1}}
\newcommand{\Tr}{\mathrm{Tr}}
\newcommand{\e}{\varepsilon}
\newcommand{\I}{{\rm i}}
\newcommand{\D}{\mathrm{d}}
\newcommand{\C}{\mathbb{C}}
\newcommand{\R}{\mathbb{R}}
\newcommand{\N}{\mathbb{N}}
\newcommand{\Z}{\mathbb{Z}}
\renewcommand{\Re}{\mathrm{Re}}

\DeclareMathOperator\supp{supp}

\newtheorem{prop}{Proposition}[section]
\newtheorem{thm}[prop]{Theorem}
\newtheorem{lem}[prop]{Lemma}
\newtheorem{defin}[prop]{Definition}
\newtheorem{cor}[prop]{Corollary}

\newtheorem{cla}[prop]{Claim}

\newtheorem{rem}[prop]{Remark}
\newenvironment{remark}{\begin{rem}\normalfont}{\end{rem}}

\title{Brownian Motions with One-Sided Collisions: The Stationary Case}
\date{February 5th, 2015}
\author{Patrik L.\ Ferrari\thanks{Institute for Applied Mathematics, Bonn University, Endenicher Allee 60, 53115 Bonn, Germany. E-mail: {\tt ferrari@uni-bonn.de}} \and
Herbert Spohn\thanks{Zentrum Mathematik, TU M\"unchen, Boltzmannstrasse 3, D-85747 Garching, Germany. E-mail: {\tt spohn@ma.tum.de}} \and
Thomas Weiss\thanks{Zentrum Mathematik, TU M\"unchen, Boltzmannstrasse 3, D-85747 Garching, Germany. E-mail: {\tt tweiss@ma.tum.de}}}

\begin{document}

\maketitle
\sloppy
\begin{abstract}
We consider an infinite system of Brownian motions which interact through a given Brownian motion being reflected from its left neighbor. Earlier we studied this system for deterministic periodic initial configurations. In this contribution we consider initial configurations distributed according to a Poisson point process with constant intensity, which makes the process space-time stationary. We prove convergence to the Airy process for stationary the case. As a byproduct we obtain a novel representation of the finite-dimensional distributions of this process. Our method differs from the one used for the TASEP and the KPZ equation by removing the initial step only after the limit $t\to\infty$. This leads to a new universal cross-over process. 
\end{abstract}
\newpage

\section{Introduction}\label{sec1}
\setcounter{equation}{0}
We will study an infinite system of interacting Brownian motions with \mbox{$x_n(t) \in \mathbb{R}$}, $n\in\Z$, denoting the position of the $n$-th Brownian particle on the real line at time $t$. Initially the positions are ordered as $x_{n}(0) \leq x_{n+1}(0)$, with the convention that $x_0(0)\leq  0 < x_1(0)$.
As indicated in the title, particle $n+1$ interacts through a steep, narrowly supported potential with its left neighbor,
$n$, only. In the limit of zero support, the singular limit studied in our contribution,  this interaction amounts to Brownian motion $n+1$ being reflected from Brownian motion $n$.
 A mathematical definition will be given below. Under this dynamics the order is preserved,
\begin{equation}\label{1.1}
x_{n}(t) \leq x_{n+1}(t)
\end{equation}
for all times $t \geq 0$. In previous work we investigate the case  of initial conditions with equal spacing, $x_n(0) = n$  \cite{FSW13}. Another natural initial condition is such  to have the process space-time stationary, which is accomplished by assuming that $\{x_{n}(0), n \in \mathbb{Z}\}$ is a Poisson point process with uniform intensity which, without loss of generality, can be taken as $1$. Then $\{x_{n}(t), n \in \mathbb{Z}\}$ is again an intensity 1 Poisson point process.

Our interest are the fluctuations of $x_n(t)$  for large $t$ and $n$. To understand their properties one first has to find out
how an  initially small perturbation close to the origin propagates in time. This path is known as characteristic.
In our model, because of the one-sided collisions, the characteristic turns out to be a straight line with velocity 1.
If $n = \lfloor \vartheta t\rfloor$,  $\lfloor \cdot\rfloor$
denoting integer part, then for $\vartheta \neq 1$ only the randomness of the initial conditions plays a role and the fluctuations of $x_n(t)$ will be Gaussian asymptotically on the $t^{1/2}$ scale.
However close to the characteristic, i.e., $n = \lfloor  t + rt^{2/3}\rfloor$ with $r = \mathcal{O}(1)$, one observes non-Gaussian fluctuations in the $t^{1/3}$ scale, the properties of which will be analysed in great detail in this contribution. With our methods
we can handle the stochastic process in $r$ at fixed $t$. Two-time properties along the characteristic are known to be difficult.
For example, a long-standing problem is to obtain the joint distribution of $x_{\lfloor t\rfloor}(t),x_{\lfloor 2t\rfloor}(2t)$
for large $t$. At the time of writing Johansson reports on asymptotic results for the model studied in this paper \cite{Jo15}.

Our results are closely linked to the one-dimensional Kardar-Parisi-Zhang (KPZ) equation \cite{KPZ86} with stationary initial data. KPZ is a stochastic PDE for a height function $h(x,t) \in \mathbb{R}$ and reads
\begin{equation}\label{1.3}
\partial_t h = \tfrac{1}{2}(\partial_x h)^2 + \tfrac{1}{2}\partial_x^2h  + W
\end{equation}
with $W$ space-time white noise. As written the equation is only formal, but a precise mathematical meaning has been given \cite{ACQ10,Hai11}. As random initial data $h(x,0)$ we choose the statistics of two-sided Brownian  motion with constant drift $b$.
The dynamics is stationary in the sense that $x \mapsto h(x,t) - h(0,t)$ is again two-sided Brownian motion with drift $b$ \cite{FQ14}.
Very recently, Borodin \textit{et al.} \cite{BCFV14} succeeded in writing down reasonably concise formulas for the distribution of
$h(x,t)$, confirming the prior replica computation \cite{IS13}. Through an intricate asymptotic analysis they establish (Theorem 2.17 of~\cite{BCFV14}) that
in distribution,  for fixed $r$,
\begin{equation}\label{1.4}
\lim_{t \to \infty}t^{-1/3}\big(h(-2bt + 2r t^{2/3},2t) + (\tfrac{1}{12} + b^2)t - 2brt^{2/3}\big) = \mathcal{A}_{\rm stat}(r)\,,
\end{equation}
where $\mathcal{A}_{\rm stat}$ denotes
the Airy process corresponding to stationary initial data.
In spirit one should think of $x_n(t)$ as $h(x,t)$ with $x$ being a continuum version of the discrete particle label
$n$.
More precisely, as one of our results we will establish in Theorem~\ref{thmAsymp0} that
\begin{equation}\label{1.5}
\lim_{t \to \infty}t^{-1/3}\big(x_{\lfloor  t + 2rt^{2/3}\rfloor}(t) -2t -2rt^{2/3}\big) = \mathcal{A}_{\rm stat}(r),
\end{equation}
which is the immediate analogue of \eqref{1.4}. In fact, convergence is proved in the sense of finite-dimensional distribution, not only for the one-point distribution.

Similar results have been obtained earlier for the stationary PNG model \cite{PS02b} and for the stationary TASEP \cite{FS05a}. For the latter, the full stochastic process in $r$  has been worked out \cite{BFP09}. The expression we obtain for the joint distributions of $\mathcal{A}_{\rm stat}$, see Definition~\ref{DefAstat}, is new and differs from the one in \cite{BFP09}.

For several reasons we believe that it is of interest to add a third model to the list of KPZ type models
with stationary initial conditions. Obviously, the universality hypothesis is further strengthened. More importantly
our model provides a bridge to diffusion processes with one-sided interaction as discussed in \cite{SS14}.
Besides we also have to develop a method different from the previous ones. As in the case of PNG and TASEP, one cannot
study the stationary initial conditions directly. One has to start from a step, in our case meaning that to the right of $0$
the Poisson point process has density $1$, while to the left it has density $\rho$, $\rho<1$. Surprisingly,
as for PNG and TASEP, by a Burke type theorem the left-half system can be replaced by a boundary condition for $x_0(t)$ and, in fact, only the right-half system with labels $\{n \geq 0\}$ has to be considered. For PNG and TASEP the limit $\rho \to 1$
has been accomplished for fixed $t$, while here we first take the limit $t\to\infty$ at step size $1 -\rho= t^{-1/3}\delta$. This leads us to a novel transition process, see Theorem~\ref{thmAsymp}. The stationary case, $\delta = 0$, is then reached through a careful analytic continuation.\bigskip\\

\noindent \textbf{Acknowledgments}.
The work of P.L.~Ferrari is supported by the German Research Foundation via the SFB 1060--B04 project.
The work of H.~Spohn is supported by the Fondation Science Math\'{e}matiques de Paris.
The work of T.~Weiss is supported by the German Research Foundation  project \mbox{SP181/29-1}.

\section{Main results}\label{SectResults}
To state our result we introduce the rescaled process
\begin{equation}\label{eqScaledProcessOriginal}
	r\mapsto X_t(r) = t^{-1/3}\big(x_{\lfloor t+2rt^{2/3}\rfloor}(t)-2t-2rt^{2/3} \big)\,
\end{equation}
and define the limit process $r\mapsto \mathcal{A}_{\rm stat}(r)$.
\begin{defin}[Airy$_{\rm stat}$ process]\label{DefAstat}
Let $P_s$ be the projection operator on $[s,\infty)$ and $\bar{P}_s=\Id-P_s$ the one on $(-\infty,s)$. Set
\begin{equation}\label{DefV}
	V_{r_1,r_2}(s_1,s_2)=\frac{e^{-\frac{(s_2-s_1)^2}{4(r_2-r_1)}}}{\sqrt{4\pi (r_2-r_1)}},
\end{equation}
and define
\begin{equation}
	\mathcal{P}=\Id-\bar{P}_{s_1}V_{r_1,r_2}\bar{P}_{s_2}\cdots V_{r_{m-1},r_m}\bar{P}_{s_m}V_{r_m,r_1},
\end{equation}
as well as an operator $K$ with integral kernel
\begin{equation}
	K(s_1,s_2)=e^{r_1(s_2-s_1)}\int_0^\infty\D x\, \Ai(r_1^2+s_1+x)\Ai(r_1^2+s_2+x).
\end{equation}
Further, define the functions
\begin{equation}\begin{aligned}
	\mathcal{R}&=s_1+e^{\frac{2}{3}r_1^3}\int_{s_1}^\infty \D x\int_x^\infty \D y\, \Ai(r_1^2+y)e^{r_1y},\\
	f^*(s)&=-e^{-\frac{2}{3}r_1^3}\int_s^\infty \D x\, \Ai(r_1^2+x)e^{-r_1x},\\
	g(s)&=1-e^{\frac{2}{3}r_1^3}\int_s^\infty \D x\, \Ai(r_1^2+x)e^{r_1x}.
\end{aligned}\end{equation}
With these definitions, set
\begin{equation} G_m(\vec{r},\vec{s})=\mathcal{R}-\left\langle(\Id-\mathcal{P}K)^{-1}\left(\mathcal{P}f^*+\mathcal{P}KP_{s_1}\mathbf{1}+(\mathcal{P}-P_{s_1})\mathbf{1}\right),g\right\rangle,
\end{equation}
where $\langle\cdot,\cdot\rangle$ denotes the inner product on $L^2(\R)$.
Then, the \emph{Airy$_{\rm{stat}}$ process}, $\mathcal{A}_{\rm stat}$, is the process with $m$-point joint distributions at $r_1<r_2<\dots<r_m$ given by
\begin{equation}\label{eqAiryDef}
	\Pb\bigg(\bigcap_{k=1}^m\{\mathcal{A}_{\rm stat}(r_k)\leq s_k\}\bigg) = \sum_{i=1}^m\frac{\D}{\D s_i}\left(G_m(\vec{r},\vec{s})\det\left(\Id-\mathcal{P}K\right)_{L^2(\R)}\right).
\end{equation}
\end{defin}
We can now state our main result.
\begin{thm}\label{thmAsymp0}
In the sense of finite-dimensional distributions,
\begin{equation}\label{eqX0limit}
	\lim_{t\to\infty}X_t(r)\stackrel{d}{=}\mathcal{A}_{\rm stat}(r).
\end{equation}
\end{thm}
 \begin{remark} The joint distributions of the \emph{Airy$_{\rm{stat}}$ process} were first obtained in~\cite{BFP09}, see Definition 1.1 and Theorem 1.2 therein. In Definition \ref{DefAstat} we state an alternative formula for the joint distributions of  the Airy$_{\rm stat}$ process. The main difference between the two formulas is that in~\cite{BFP09} the joint distributions are given in terms of a Fredholm determinant on $L^2(\{1,\ldots,m\}\times\R)$, while here we have a Fredholm determinant on $L^2(\R)$. A similar twist was already visible in~\cite{PS02} and has been generalized in~\cite{BCR13}.
\end{remark}

Since $\{x_{n}(t), n \in \mathbb{Z}\}$ is a Poisson point process,
the process $X_t(r) - X_t(0)$ is a scaled Poisson jump process up to a linear part and
\begin{equation}
 \lim_{t \to \infty}\big(X_t(r)-X_t(0)\big)\stackrel{d}{=}B(2r)\,.
\end{equation}
Hence the limit process $\mathcal{A}_{\rm stat}(r) -\mathcal{A}_{\rm stat}(0)$ must also have the statistics of two-sided Brownian motion, a property which is not so easily inferred from our formulas  in Definition \ref{DefAstat}. But we will provide a direct proof of this fact  in Section~\ref{SectGaussianIncr}. Note that $X_t(0)$ and $B(2r)$ are not independent.

As already familiar from other models in the KPZ universality class~\cite{BFP09,BCFV14,FS05a,SI04}, the proof of Theorem~\ref{thmAsymp0} proceeds via a sequence of approximating initial conditions. Firstly we consider the case where $x_0(0)=0$ and assume that the particles on $\R_+$ is a Poisson process with intensity $\lambda>0$ and on $\R_-$ is a Poisson process with intensity $\rho>0$. In other words, $x_n(0)=\zeta_n$, $n\in\Z$, with
\begin{equation}\label{statModel}
 \begin{aligned}
  \zeta_0&=0,\\
  \zeta_n-\zeta_{n-1}&\sim\exp(\lambda),\quad &\text{for } n>0,\\
  \zeta_n-\zeta_{n-1}&\sim\exp(\rho), &\text{for } n\leq0.
 \end{aligned}
\end{equation}
As explained in Lemma~\ref{lemma7.1}, setting $\zeta_0=0$ will induce a difference of order one as compared to the case considered in Theorem~\ref{thmAsymp0}. In the scaling limits such differences are irrelevant. Thus it is enough to prove Theorem~\ref{thmAsymp0} for the initial conditions \eqref{statModel} with $\lambda = 1 =\rho$.
\textit{In the sequel $x_n(t)$  always refers to the initial conditions \eqref{statModel}, in such a way that the choice of the parameters $\lambda,\rho$ can be inferred from the context}. We obtain the fixed time multi-point distributions of the system $\{x_n(t),n\in\N_0\}$ in terms of a Fredholm determinant in the case $\lambda >\rho$.
The restriction to non-negative integers comes from Burke's theorem by which the  particles with $n <0 $
can be replaced by choosing $x_0(t)$ as a Brownian motion with drift $\rho$.

\begin{prop}\label{propKernel}
Let $\lambda>\rho>0$. For any finite subset $S$ of $\N_0$, it holds
\begin{equation}\label{eq33}
\Pb\bigg(\bigcap_{n\in S} \{x_n(t)\leq a_n\}\bigg)=\bigg(1+\frac{1}{\lambda-\rho}\sum_{k\in S}\frac{\D}{\D a_k}\bigg)\det(\Id-\chi_a \mathcal{K} \chi_a)_{L^2(S\times \R)},
\end{equation}
where $\chi_a(n,\xi)=\Id(\xi>a_n)$. The kernel $\mathcal{K}$ is given by
\begin{equation}\label{eqKtflat}
\begin{aligned}	\mathcal{K}(n_1,\xi_1;n_2,\xi_2)&=-\phi_{n_1,n_2}(\xi_1,\xi_2)\Id_{(n_2>n_1)}+\widetilde{\mathcal{K}}(n_1,\xi_1;n_2,\xi_2)\\&\quad+(\lambda-\rho)\mathpzc{f}(n_1,\xi_1)\mathpzc{g}(n_2,\xi_2).
\end{aligned}
\end{equation}
where
\begin{equation}\begin{aligned}
	\phi_{0,n_2}(\xi_1,\xi_2)&=\rho^{-n_2}e^{\rho \xi_2},&\quad& \text{for } n_2\geq0,\\
	\phi_{n_1,n_2}(\xi_1,\xi_2)&=\frac{(\xi_2-\xi_1)^{n_2-n_1-1}}{(n_2-n_1-1)!}\Id_{\xi_1\leq \xi_2},&& \text{for } 1\leq n_1<n_2,
	\end{aligned}\end{equation}
and	
	\begin{equation}\begin{aligned}	\widetilde{\mathcal{K}}(n_1,\xi_1;n_2,\xi_2)&=\frac{1}{(2\pi\I)^2}\int_{\I\R-\e}\D w\,\oint_{\Gamma_0}\D z\,\frac{e^{{t} w^2/2+\xi_1w}}{e^{{t} z^2/2+\xi_2 z}}\frac{(-w)^{n_1}}{(-z)^{n_2}}\frac{1}{w-z},\\
	\mathpzc{f}(n_1,\xi_1)&=\frac{1}{2\pi\I}\int_{\I\R-\e}\D w\,\frac{e^{{t} w^2/2+\xi_1w}(-w)^{n_1}}{w+\lambda},\\
	\mathpzc{g}(n_2,\xi_2)&=\frac{1}{2\pi\I}\oint_{\Gamma_{0,-\rho}}\D z\,\frac{e^{-{t} z^2/2-\xi_2 z}(-z)^{-n_2}}{z+\rho},
\end{aligned}\end{equation}
for any fixed $0<\e<\lambda$.
\end{prop}

Notice that this result holds for $\lambda>\rho$ only and not for the most interesting case $\lambda=\rho$. The latter can be accessed through a careful analytic continuation of the formulas. One of the novelty of this paper is to perform the analytic continuation \emph{after} the scaling limit. This allows us to discover a new process, called \emph{finite-step Airy$_{\rm{stat}}$ process},  describing the large time limit close to stationarity (actually, one still needs to take care of the random shift of $x_0(0)$, which is however irrelevant as it goes to zero after scaling in the large time limit). As before, this process is defined
through its finite-dimensional distributions.
\begin{defin}[Finite-step Airy$_{\rm stat}$ process]\label{DefAiryStat}The finite-step Airy$_{\rm{stat}}$ process with parameter $\delta>0$, $\mathcal{A}_{\rm stat}^{(\delta)}$, is the process with m-point joint distributions at $r_1<r_2<\dots<r_m$ given by
\begin{equation}\label{eq3.6}
	\Pb\bigg(\bigcap_{k=1}^m\{\mathcal{A}_{\rm stat}^{(\delta)}(r_k)\leq s_k\}\bigg) = \bigg(1+\frac{1}{\delta}\sum_{i=1}^m\frac{\D}{\D s_i}\bigg)\det\left(\Id-\chi_s K^\delta\chi_s\right)_{L^2(\{r_1,\dots,r_m\}\times\R)},
\end{equation}
where $\chi_s(r_k,x)=\Id(x>s_k)$ and the kernel $K^\delta$ is defined by
\begin{equation}
K^\delta(r_1,s_1;r_2,s_2)=-V_{r_1,r_2}(s_1,s_2)\Id_{(r_1<r_2)}+K_{r_1,r_2}(s_1,s_2)+\delta f_{r_1}(s_1)g_{r_2}(s_2).
\end{equation}
Here, $V_{r_1,r_2}$ is defined as in (\ref{DefV}), and
\begin{equation}\label{eqKernelDef}
\begin{aligned}
	K_{r_1,r_2}(s_1,s_2)&=\frac{e^{\frac{2}{3}r_2^3+r_2s_2}}{e^{\frac{2}{3}r_1^3+r_1s_1}}\int_0^\infty\D x\, e^{-x(r_1-r_2)}\Ai(r_1^2+s_1+x)\Ai(r_2^2+s_2+x)\\
		f_{r_1}(s_1)&=1-e^{-\frac{2}{3}r_1^3-r_1s_1}\int_0^\infty\D x\, \Ai(r_1^2+s_1+x)e^{-r_1x}\\
		g_{r_2}(s_2)&=e^{ \delta^3/3+r_2\delta^2-s_2\delta}-e^{\frac{2}{3}r_2^3+r_2s_2}\int_0^\infty\D x\, \Ai(r_2^2+s_2+x)e^{(\delta+r_2)x}.
\end{aligned}
\end{equation}
\end{defin}

As mentioned already above, we are going to take the limit to stationarity after the long time limit. However, in general, the limits $t\to\infty$ and $\lambda-\rho\downarrow 0$ do not commute. Therefore we have to consider $\lambda-\rho>0$ (to be able to apply Proposition~\ref{propKernel}), but vanishing with a tuned scaling exponent as $t\to\infty$, a critical scaling. We set $\lambda-\rho=\delta t^{-1/3}$ for $\delta>0$. As will be proven with this choice the limit $t\to\infty$ commutes with $\delta\downarrow 0$.

Such considerations lead  naturally to define the rescaled process as
\begin{equation}\label{eqScaledProcess}
r\mapsto X^{(\delta)}_t(r) = t^{-1/3}\Big(x^{(1-t^{-1/3}\delta)}_{\lfloor t+2rt^{2/3}\rfloor}(t)-2t-2rt^{2/3}\Big)\,,
\end{equation}
where the superscript of $x$ indicates $\lambda = 1$ and  $\rho = 1-t^{-1/3}\delta$.

The second main result of our paper is the description of the joint distributions of the rescaled process in the long time limit.
\begin{thm}\label{thmAsymp}
For every $\delta>0$, the rescaled process \eqref{eqScaledProcess} converges to the \emph{finite-step Airy$_{\rm{stat}}$ process}
\begin{equation}
	\lim_{t\to\infty}X^{(\delta)}_t(r)\stackrel{d}{=}\mathcal{A}^{(\delta)}_{\rm stat}(r),
\end{equation}
in the sense of finite-dimensional distributions.
\end{thm}

\newpage
\section{Semi-infinite initial conditions}\label{SectIC}
\subsection{Well-definiteness}
Consider the initial conditions stated in (\ref{statModel}). First we show that the system with infinitely many particles is well-defined. For that purpose we use the Skorokhod representation~\cite{Sko61,AO76} to define the reflected Brownian motions. This representation is the following: the process $x(t)$, driven by the Brownian motion $B(t)$, starting from $x(0)\in\R$ and being reflected at some continuous function $f(t)$ with $f(0)<x(0)$ is defined as:
\begin{equation}\begin{aligned}
x(t)&=x(0)+B(t)-\min\big\{0,\inf_{0\leq s \leq t}(x(0)+B(s)-f(s))\big\} \\
&= \max\big\{x(0)+B(t),\sup_{0\leq s \leq t}(f(s)+B(t)-B(s))\big\}.
\end{aligned}\end{equation}

Let $B_n$, $n\in\Z$, be independent standard Brownian motions starting at $0$ and define the random variables
\begin{equation}\label{eq2.2}
Y_{k,n}(t)=\sup_{0\leq s_{k+1}\leq \ldots\leq s_m\leq t}\sum_{i=k}^n (B_i(s_{i+1})-B_i(s_i))
\end{equation}
with the convention $s_{k}=0$ and $s_{n+1}=t$.
We will define the system $\{x_n(t),n\in\Z\}$ as the limit of half-infinite systems $\{x_n^{(M)}(t),n\geq -M\}$ as $M\to\infty$, where
\begin{equation}\label{eq3a}
x_n^{(M)}(t)=\max_{k\in [-M,n]}\{Y_{k,n}(t)+\zeta_k\},\quad n\geq-M.
\end{equation}
Notice that these processes indeed satisfy the Skorokhod equation,
\begin{equation}
x_n^{(M)}(t)= \max\big\{\zeta_n+B_n(t),\sup_{0\leq s \leq t}(x_{n-1}^{(M)}(s)+B_n(t)-B_n(s))\big\},
\end{equation}
for $n>-M$, while the leftmost process is simply
\begin{equation}
 x_{-M}^{(M)}(t)=\zeta_{-M}+B_{-M}(t).
\end{equation}
Thus as desired $x_n^{(M)}(t)$ is a Brownian motion starting from $\zeta_n$ and reflected off by $x_{n-1}^{(M)}$ for $n>-M$. The representation (\ref{eq3a}) can be sees as a zero-temperature case of the O'Connell-Yor semi-discrete directed polymer~\cite{OCY01} with appropriate boundary conditions (see discussion at the end of this section).

Next we show strong converge of the system $\{x_n^{(M)}(t),n\geq -M\}$ to the processes we are studying.
\begin{prop}\label{propMconv}
Let us define
\begin{equation}
x_n(t)=\max_{k\leq n}\{Y_{k,n}(t)+\zeta_k\}.
\end{equation}
Then, for any $T>0$,
\begin{equation}
	\lim_{M\to\infty}\sup_{t\in[0,T]}|x_n^{(M)}(t)-x_n(t)|=0,\quad \mathrm{a.s.}
\end{equation}
as well as
\begin{equation}\label{eqUniBound}
	\sup_{t\in[0,T]}|x_n(t)|<\infty,\quad \mathrm{a.s.}.
\end{equation}
\end{prop}

For the proof of this proposition we first need the following concentration inequality, which is Proposition 2.1 of \cite{ledDIL}:
\begin{prop}\label{concIn}
For each $T>0$ there exists a constant $C>0$ such that for all \mbox{$k<m$}, $\delta>0$,
\begin{equation}
	\Pb\bigg(\frac{Y_{k,m}(T)}{2\sqrt{(m-k+1)T}}\geq 1+\delta\bigg)\leq Ce^{-(m-k+1)\delta^{3/2}/C}.
\end{equation}
\end{prop}

\begin{proof}[Proof of Proposition~\ref{propMconv}]
Let us define an auxiliary system of processes, which we will use later in proving the Burke property, by
\begin{equation}
 \widetilde{x}_{-M}^{(M)}(t)=\zeta_{-M}+B_{-M}(t)+\rho t,
\end{equation}
and
\begin{equation}
\widetilde{x}_n^{(M)}(t)= \max\big\{\zeta_n+B_n(t),\sup_{0\leq s \leq t}(\widetilde{x}_{n-1}^{(M)}(s)+B_n(t)-B_n(s))\big\}
\end{equation}
for $n>-M$. This system differs from $x_n^{(M)}(t)$ just in the drift of the leftmost particle, which of course influences all other particles as well (the choice of the extra drift is because the system with infinite many particles in $\R_-$ generates a drift $\rho$). This system of particles satisfies
\begin{equation}\label{eq3}
\widetilde{x}_n^{(M)}(t)=\max\big\{\widetilde{Y}_{-M,n}(t)+\zeta_{-M},\max_{k\in [-M+1,n]}\{Y_{k,n}(t)+\zeta_k\}\big\},
\end{equation}
with
\begin{equation}
\widetilde{Y}_{k,n}(t)=\sup_{0\leq s_{k+1}\leq \ldots\leq s_m\leq t}\Big(\rho s_{k+1}+\sum_{i=k}^n (B_i(s_{i+1})-B_i(s_i))\Big).
\end{equation}
Also, we have the inequalities
\begin{equation}\label{eq2.16}
 Y_{k,n}(t)\leq\widetilde{Y}_{k,n}(t)\leq Y_{k,n}(t)+\rho t.
\end{equation}

Consider the event
\begin{equation}\label{eq3.15}
	A_M:=\{Y_{-M,n}(T)+\zeta_{-M}+\rho T\geq-M/2\}\cup\{Y_{n,n}(T)+\zeta_n\leq-M/2\}.
\end{equation}
We can deduce exponential decay of $\Pb(A_M)$ in $M$ from combining the exponential tails of $Y_{n,n}$ and $\zeta_n$ with Proposition~\ref{concIn}, using $\delta=1$ and elementary inequalities. In particular $\sum_{M=1}^\infty\Pb(A_M)<\infty$, so by Borel-Cantelli, $A_M$ occurs only finitely many times almost surely. This means, that a.s.\  there exists a $M_T$, such that for all $M\geq M_T$:
\begin{equation}\begin{aligned}
	Y_{-M,n}(T)+\zeta_{-M}+\rho T&<-M/2 \quad\text{and}\quad
	Y_{n,n}(T)+\zeta_n&>-M/2.
\end{aligned}\end{equation}
Consequently, $Y_{-M,n}(T)+\zeta_{-M}<Y_{n,n}(T)+\zeta_n$ for all $M\geq M_T$ and therefore
\begin{equation}
	x_n(T)=x_n^{(M_T)}(T).
\end{equation}
Moreover, applying \eqref{eq2.16}, gives
\begin{equation}
 Y_{-M,n}(t)+\zeta_{-M}\leq\widetilde{Y}_{-M,n}(t)+\zeta_{-M}<Y_{n,n}(T)+\zeta_n,
\end{equation}
resulting in
\begin{equation}
	\widetilde{x}_n^{(M_T)}(T)=x_m^{(M_T)}(T).
\end{equation}

Repeating the same argument, we see that for every $t\in[0,T]$ there exists $M_t$ such that $x_n(t)=x_n^{(M_t)}(t)=\widetilde{x}_n^{(M_t)}(t)$. Applying Lemma~\ref{lemMaxPath} then gives $x_n(t)=x_n^{(M_T)}(t)=\widetilde{x}_n^{(M_T)}(t)$ for every $t\in[0,T]$. This settles the convergence.

To see \eqref{eqUniBound}, which is equivalent to $\sup_{t\in[0,T]}|x_m^{(M_T)}(t)|<\infty$, we apply the bound
\begin{equation}
	|Y_{k,n}(t)|\leq\sum_{i=k}^n\Big(\sup_{0\leq s\leq t}B_i(s)-\inf_{0\leq s\leq t}B_i(s)\Big)<\infty.
\end{equation}
\end{proof}

\begin{lem}\label{lemMaxPath}
Consider $0\leq t_1\leq t_2$ and $m$, $M_{t_1}$, $M_{t_2}$ such that
\begin{equation}\label{eq10a}
	x_m(t_i)=x_m^{(M_{t_i})}(t_i)=\widetilde{x}_m^{(M_{t_i})}(t_i),\quad\text{for }i=1,2.
\end{equation}
Then
\begin{equation}\label{eq11}
	x_m(t_1)=x_m^{(M_{t_2})}(t_1)=\widetilde{x}_m^{(M_{t_2})}(t_1).
\end{equation}
\end{lem}
\begin{proof}
This is a straightforward generalization of Lemma 3.2~\cite{FSW13}.
\end{proof}

\subsection{Burke's property}
We establish a useful property which will allow us to study our system of interacting Brownian motions through a system with a left-most Brownian particle.
\begin{prop}\label{propBM}
For each $n\leq0$, the process
\begin{equation}\label{eq2.9}
 x_n(t)-\zeta_n-\rho t
\end{equation}
is a standard Brownian motion.
\end{prop}

\begin{remark}\label{RemarkBurke}
Proposition~\ref{propBM} allows us to restrict our attention to the half-infinite system. In fact, conditioned on the path of $x_0$, the systems of particles $\{x_n(t),n<0\}$ and $\{x_n(t),n>0\}$ are independent, as it is clear by the definition of the system. Then (\ref{eq2.9}) implies that the law of \mbox{$\{x_n(t), n>0\}$} is the same as the one obtained replacing the infinitely many particles $\{x_m(t), m\leq 0\}$ with a single Brownian motion $x_0(t)$ which has a drift $\rho$. This property will be used to derive our starting result, Proposition~\ref{propKernel}.
\end{remark}

\begin{proof}[Proof of Proposition~\ref{propBM}]
First notice that
\begin{equation}
 \widetilde{x}_{-M}^{(M)}(t)-\zeta_{-M}-\rho t,
\end{equation}
is a Brownian motion. Now assume $\widetilde{x}_{n-1}^{(M)}(t)-\zeta_{n-1}-\rho t$ is a Brownian motion. By definition,
\begin{equation}
\widetilde{x}_n^{(M)}(t)-\zeta_{n-1}= \max\big\{\zeta_n-\zeta_{n-1}+B_n(t),\sup_{0\leq s \leq t}(\widetilde{x}_{n-1}^{(M)}(s)-\zeta_{n-1}+B_n(t)-B_n(s))\big\},
\end{equation}
which allows us to apply Proposition~\ref{propNeil}, i.e., we have that
\begin{equation}
 \widetilde{x}_n^{(M)}(t)-\zeta_{n-1}-(\zeta_n-\zeta_{n-1})-\rho t=\widetilde{x}_n^{(M)}(t)-\zeta_n-\rho t
\end{equation}
is a Brownian motion. Since $\widetilde{x}_n^{(M_T)}(t)=x_n(t)$ the proof is completed.
\end{proof}

It is clear, that in the case $\lambda=\rho$ the process \eqref{eq2.9} is a Brownian motion for $n>0$, too, i.e.,  the system is stationary in $n$. We also have stationarity in $t$, in the sense that for each $t\geq0$ the random variables \mbox{$\{x_n(t)-x_{n-1}(t),n\in\Z\}$} are independent and distributed exponentially with parameter $\rho$. The following result is a small modification of Theorem 2 in~\cite{OCY01}.

\begin{prop}[Burke's theorem for Brownian motions]\label{propNeil}
Fix $\rho>0$ and let $B(t)$, $C(t)$ be standard Brownian motions, as well as $\zeta\sim\exp(\rho)$, independent. Define the process
\begin{equation}
D(t)=\max\big\{\zeta+C(t),\sup_{0\leq s \leq t}(B(s)+\rho s+C(t)-C(s))\big\}.
\end{equation}
Then
\begin{equation}\label{eq2.29}
 D(t)-\zeta-\rho t
\end{equation}
is distributed as a standard Brownian motion.
\end{prop}
\begin{proof}
Extend the processes $B(t)$, $C(t)$ to two-sided Brownian motions indexed by $\R$. Defining
\begin{equation}
 q(t)=\sup_{-\infty<s\leq t}\{B(t)-B(s)+C(t)-C(s)-\rho(t-s)\}
\end{equation}
and
\begin{equation}
 d(t)=B(t)+q(0)-q(t),
\end{equation}
we can apply Theorem 2~\cite{OCY01}, i.e., $d(t)$ is a Brownian motion. Now,
\begin{equation}
 q(0)=\sup_{s\leq0}\{-B(s)-C(s)+\rho s\}\stackrel{d}{=}\sup_{s\geq0}\{\sqrt{2}B(s)-\rho s\}\stackrel{d}{=}\sup_{s\geq0}\big\{B(s)-\frac{\rho}{2} s\big\},
\end{equation}
so by Lemma~\ref{lemBMExp} it has exponential distribution with parameter $\rho$. As it is independent of the processes $\{B(t),C(t),t\geq0\}$ we can write $q(0)=\zeta$. Dividing the supremum into $s<0$ and $s\geq0$ we arrive at:
\begin{equation}\begin{aligned}
 -d(t)&=q(t)-B(t)-q(0)\\
 &=\max\Big\{C(t)-\rho t,\sup_{0\leq s\leq t}\{-B(s)+C(t)-C(s)-\rho(t-s)\}-\zeta\Big\},
\end{aligned}\end{equation}
which is \eqref{eq2.29} up to a sign flip of $B(s)$.
\end{proof}

\begin{lem}\label{lemBMExp}
 Fix $\rho>0$ and let $B(t)$ be a standard Brownian motion. Then
 \begin{equation}
  \sup_{s\geq0}(B(s)-\rho s)\sim\exp(2\rho).
 \end{equation}
\end{lem}
\begin{proof}
The random variable
 \begin{equation}
  \sup_{0\leq s\leq t}(B(s)-\rho s)
 \end{equation}
is distributed as a Brownian motion with drift $-\rho$ started at zero and being reflected (upwards) at zero, at time $t$. As $t\to\infty$, this converges to the stationary distribution of this process, which is the exponential distribution with parameter $2\rho$.
\end{proof}

From a stochastic analysis point of view, the system \mbox{$\{x_n(t),n\geq0\}$} satisfies
\begin{equation}\begin{aligned}
 x_n(t)&=\zeta_n+B_n(t)+L^n(t),\quad\text{for } n\geq1,\\
 x_0(t)&=\widetilde{B}_0(t)+\rho t.
 \end{aligned}\end{equation}
Here $L^n$, $n\geq2$, are continuous non-decreasing processes increasing only when $x_n(t)=x_{n-1}(t)$. In fact, $L^n$ is twice the semimartingale local time at zero of $x_n-x_{n-1}$.
Notice that $\widetilde{B}_0(t)$ is a standard Brownian motion independent of $\{\zeta_n,B_n(t),n\geq1\}$, but not equal to $B_0(t)$.

\subsection{Last passage percolation interpretation}\label{SectLPP}
One can also view the system \mbox{$\{x_n(t),n\geq0\}$} as a model for last passage percolation (or zero-temperature semi-discrete directed polymer). We assign random background weights on the set $\{\R_+\times \N_0\}$ in the following way:
\begin{itemize}
 \item White noises $\D B_n$ on the lines $\R_+\times \{n\}$ for $n\geq1$,
 \item White noise $\D \widetilde{B}_0$ plus a Lebesgue measure of density $\rho$ on the line $\R_+\times\{0\}$, and
 \item Dirac measures of magnitude $\zeta_n-\zeta_{n-1}$ on $(0,n)$ for $n\geq1$.
\end{itemize}

\begin{figure}
\begin{center}
 \psfrag{0}[lb]{$0$}
 \psfrag{1}[lb]{$1$}
 \psfrag{2}[lb]{$2$}
 \psfrag{3}[lb]{$3$}
 \psfrag{4}[lb]{$4$}
 \psfrag{s0}[lb]{$s_0$}
 \psfrag{s1}[lb]{$s_1$}
 \psfrag{s2}[lb]{$s_2$}
 \psfrag{s3}[lb]{$s_3$}
 \psfrag{t}[lb]{$t$}
\includegraphics[height=5cm]{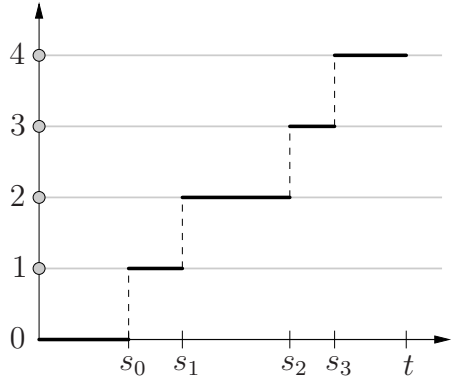}
 \caption{A path $\pi\in\Pi(0,0;t,4)$ (thick black) and the random background noise (grey).}
 \label{FigDP}
\end{center}
\end{figure}

An up-right path in $\{\R_+\times\N_0\}$ is characterized by its jumping points $s_i$ and it consists of line segments $[s_{n-1},s_n]\times \{n\}$, see Figure~\ref{FigDP}. The set of up-right paths can then be parameterized by
\begin{equation}
 \Pi(t_1,n_1;t_2,n_2) = \left\{\vec{s}\in\R^{n_2-n_1+2}|t_1=s_{n_1-1}\leq s_{n_1}\leq \dots\leq s_{n_2}=t_2\right\}.
\end{equation}
The \emph{percolation time} or \emph{weight} of a path $\vec\pi\in\Pi$ is the integral over the background weights along the path. Explicitly, we have:
\begin{equation}
 w(\vec{\pi})=\widetilde{B}_0(s_0)+\rho s_0+\sum_{i=1}^{n_2} \left((\zeta_i-\zeta_{i-1})\Id_{s_{i-1}=0} +B_i(s_{i})-B_i(s_{i-1})\right)
\end{equation}
for $n_1=0$, and for $n_1>0$,
\begin{equation}
 w(\vec{\pi})=\sum_{i=n_1}^{n_2} \left((\zeta_i-\zeta_{i-1})\Id_{s_{i-1}=0} +B_i(s_{i})-B_i(s_{i-1})\right).
\end{equation}

The \emph{last passage percolation time} is given by the supremum over all paths:
\begin{equation}
 L_{(t_1,n_1)\to(t_2,n_2)}:=\sup_{\vec{\pi}\in\Pi(t_1,n_1;t_2,n_2)}w(\vec{\pi}).
\end{equation}
The supremum is almost surely attained by a unique path $\vec\pi^*$, called the maximizer. It exists because the supremum can be rewritten as a composition of a finite maximum and a supremum of a continuous function over a compact set. Uniqueness follows from elementary properties of the Brownian measure. Most importantly, from the definition, we have
\begin{equation}
 x_n(t)=L_{(0,0)\to(t,n)}.
\end{equation}

We will use this interpretation in Section~\ref{SectAnCont}, however, it also gives some connections to different works. Our model can be seen as the semi-continuous limit of a more widely studied discrete last passage percolation model (see for example~\cite{Jo00b,Jo03b}). Moreover, our last passage percolation model is the zero temperature limit of a directed polymer model, which has been studied thoroughly in the recent past \cite{SV10,BCFV14}.

For later use we also define a version without boundary weights:
\begin{equation}
 L^{\rm step}_{(0,1)\to(t,n)}:=\sup_{\vec{\pi}\in\Pi(0,1;t,n)}\sum_{i=1}^{n} \left(B_i(s_{i})-B_i(s_{i-1})\right).
\end{equation}

\newpage
\section{Determinantal structure - Proof of Proposition~\ref{propKernel}}\label{SectDetStruct}
In order to prove Proposition~\ref{propKernel} we first start by considering the transition probability for a finite number of reflecting Brownian motions with drifts and arbitrary initial positions (Proposition~\ref{PropWarren}). Then we will set the drift of the first Brownian motion to $\rho$, see Remark~\ref{RemarkBurke}, and we will randomize the initial positions (Proposition~\ref{PropRandomIC}).

\subsection{Transition density for fixed initial positions}
Proposition~\ref{PropWarren} generalizes Proposition 4.1~\cite{FSW13}, which has been first shown in~\cite{War07}, to the case of non-zero drifts.
\begin{prop}\label{PropWarren}
The transition probability density of $N$ one-sided reflected Brownian motions with drift $\vec{\mu}$ from $\vec{x}(0)=\vec{\zeta}\in W^N$  to \mbox{$\vec{x}({t})=\vec{\xi}\in W^N$} at time ${t}$ has a continuous version, which is given as follows.
\begin{equation}\label{transDens}
	\Pb\left(\vec{x}({t})\in \D\vec{\xi}|\vec{x}(0)=\vec{\zeta}\right)=r_{t}(\vec{\zeta},\vec{\xi})\D\vec{\xi},
\end{equation}
where
\begin{equation}\label{eqrtau}
r_{t}(\vec{\zeta},\vec{\xi})=\bigg(\prod_{n=1}^Ne^{\mu_n(\xi_n-\zeta_n)-{t}\mu_n^2/2}\bigg)\det_{1\leq k,l\leq N}[F_{k,l}(\xi_{N+1-l}-\zeta_{N+1-k},{t})],
\end{equation}
and
\begin{equation}\label{eqFkl}
	F_{k,l}(\xi,{t})=\frac{1}{2\pi\I}\int_{\I\R+\mu}\D w\,e^{{t} w^2/2+\xi w}\frac{\prod_{i=1}^{k-1}(w+\mu_{N+1-i})}{\prod_{i=1}^{l-1}(w+\mu_{N+1-i})},
\end{equation}
with $\mu>-\min\{\mu_1,\dots,\mu_N\}$.

\end{prop}
\begin{proof}
We follow the proof of Proposition~8 in~\cite{War07}. The strategy is to show that the transition density satisfies three equations, the backwards equation, boundary condition and initial condition, the latter one being contained in Lemma~\ref{lemInCond}. These equations are then used for It\^o's formula to prove that it indeed is the transition density.

We start with the backwards equation and boundary condition:
\begin{align}
	\frac{\partial r_{t}}{\partial{t}}&=\sum_{n=1}^N\left(\frac{1}{2}\frac{\partial^2}{\partial\zeta_n^2}+\mu_n\frac{\partial }{\partial\zeta_n}\right)r_{t}.\label{eqBackwards}\\
	\frac{\partial r_{t}}{\partial\zeta_i}&=0,\qquad\text{whenever } \zeta_i=\zeta_{i-1},\ 2\leq i\leq N\label{eqBoundCond}
\end{align}
To see \eqref{eqBoundCond}, move the prefactor $e^{-\mu_{i}\zeta_{i}}$ inside the integral in the \mbox{$(N+1-i)$-}th row of the determinant and notice that the differential operator transforms $F_{k,l}$ into $-F_{k+1,l}$.
 Consequently, $\zeta_i=\zeta_{i-1}$ implies the $(N+1-i)$-th being the negative of the $(N+2-i)$-th row.
\eqref{eqBackwards} can be obtained by the computation
\begin{equation}
	\frac{\partial r_{t}}{\partial{t}}=\frac{1}{2}\sum_{n=1}^N\left(-\mu_n^2+e^{-\mu_n\zeta_n}\frac{\partial^2 }{\partial\zeta_n^2}e^{\mu_n\zeta_n}\right)r_{t}.
\end{equation}

Let \mbox{$f\colon W^N\rightarrow\R$} be a $C^\infty$ function, whose support is compact and has a distance of at least some $\e>0$ to the boundary of $W^N$. Define a function $F\colon(0,\infty)\times W^N\rightarrow\R$ as
\begin{equation}
	F({t},\vec{\zeta}\,)=\int_{W^N}r_{t}(\vec{\zeta},\vec{\xi})f(\vec{\xi})\,\D \vec{\xi}.
\end{equation}
The previous identities \eqref{eqBoundCond} and \eqref{eqBackwards} carry over to the function $F$ in the form of:
\begin{align}
	\frac{\partial F}{\partial\zeta_i}&=0,\qquad\text{for } \zeta_i=\zeta_{i-1},\ 2\leq i\leq N\label{eqBoundCond2}\\
	\frac{\partial F}{\partial{t}}&=\sum_{n=1}^N\left(\frac{1}{2}\frac{\partial^2}{\partial\zeta_n^2}+\mu_n\frac{\partial }{\partial\zeta_n}\right)F.\label{eqBackwards2}
\end{align}
Our processes satisfy $x_n({t})=\zeta_n+\mu_n{t}+B_n({t})+L^n({t})$, where $B_n$ are independent Brownian motions, $L^1\equiv0$ and $L^n$, $n\geq2$, are continuous non-decreasing processes increasing only when $x_n({t})=x_{n-1}({t})$. In fact, $L^n$ is twice the semimartingale local time at zero of $x_n-x_{n-1}$. Now fix some $\e>0$, $T>0$, define a process $F(T+\e-{t},\vec{x}({t}))$ for ${t}\in[0,T]$ and apply It\^o's formula:
\begin{equation}\label{eqIto}\begin{aligned}
	F\left(T+\e-{t},\vec{x}({t})\right)&=F\left(T+\e,\vec{x}(0)\right)+\int_0^{t}-\frac{\partial}{\partial s}F\left(T+\e-s,\vec{x}(s)\right)\D s
	\\&+\sum_{n=1}^N\int_0^{t}\frac{\partial}{\partial\zeta_n}F\left(T+\e-s,\vec{x}(s)\right)\D x_n(s)
	\\&+\frac{1}{2}\sum_{m,n=1}^N\int_0^{t}\frac{\partial^2}{\partial\zeta_m\partial\zeta_n}F\left(T+\e-s,\vec{x}(s)\right)\D \left\langle x_m(s),x_n(s)\right\rangle.
\end{aligned}\end{equation}
From the definition it follows that $\D x_n(t)=\mu_n\D t +\D B_n(t)+\D L^n(t)$ and
\begin{equation}
	\D \left\langle x_m(t),x_n(t)\right\rangle=\D \left\langle B_m(t),B_n(t)\right\rangle=\delta_{m,n}\D t,
\end{equation}
because continuous functions of finite variation do not contribute to the quadratic variation. Inserting the differentials, by \eqref{eqBackwards2} the integrals with respect to $\D s$ integrals cancel, which results in:
\begin{equation}\begin{aligned}
	\eqref{eqIto}&=F\left(T+\e,\vec{x}(0)\right)+\sum_{n=1}^N\int_0^{t}\frac{\partial}{\partial\zeta_n}F\left(T+\e-s,\vec{x}(s)\right)\D B_n(s)
	\\&\quad+\sum_{n=1}^N\int_0^{t}\frac{\partial}{\partial\zeta_n}F\left(T+\e-s,\vec{x}(s)\right)\D L^n(s).
\end{aligned}\end{equation}
Since the measure $\D L^n(t)$ is supported on $\{x_n(t)=x_{n-1}(t)\}$, where the spatial derivative of $F$ is zero (see (\ref{eqBoundCond2})), the last term vanishes, too. So $F\left(T+\e-{t},\vec{x}({t})\right)$ is a local martingale and, being bounded, even a true martingale. In particular its expectation is constant, i.e.:
\begin{equation}
	F(T+\e,\vec{\zeta}\,)=\E\left[F\left(T+\e,\vec{x}(0)\right)\right]=\E\left[F\left(\e,\vec{x}(T)\right)\right].
\end{equation}
Applying Lemma~\ref{lemInCond} we can take the limit $\e\to0$, leading to
\begin{equation}
	F(T,\vec{\zeta}\,)=\E\left[f\left(\vec{x}(T)\right)\right].
\end{equation}
Because of the assumptions we made on $f$ it is still possible that the distribution of $\vec{x}(T)$ has positive measure on the boundary. We thus have to show that $r_{t}(\vec{\zeta},\vec{\xi})$ is normalized over the interiour of the Weyl chamber:

Start by integrating \eqref{eqrtau} over $\xi_N\in[\xi_{N-1},\infty)$. Pull the prefactor indexed by $n=N$ as well as the integration inside the $l=1$ column of the determinant. The $(k,1)$ entry is then given by:
\begin{equation}
 \begin{aligned}
  &e^{-\mu_N\zeta_N-{t}\mu_N^2/2}\int^{\infty}_{\xi_{N-1}}\D\xi_N e^{\mu_N\xi_N}F_{k,1}(\xi_N-\zeta_{N+1-k},{t})\\
  &\quad=e^{-\mu_N\zeta_N-{t}\mu_N^2/2}  e^{\mu_N x}F_{k,2}(x-\zeta_{N+1-k},{t})\Big|^{x=\infty}_{x=\xi_{N-1}}.
 \end{aligned}
\end{equation}
The contribution from $x=\xi_{N-1}$ is a constant multiple of the second column and thus cancels out. The remaining terms are zero for $k\geq2$, since all these functions $F_{k,2}$ have Gaussian decay. The only non-vanishing term comes from $k=1$ and returns exactly $1$ by an elementary residue calculation.

The determinant can thus be reduced to the index set $2\leq k,l\leq N$. Successively carrying out the integrations of the remaining variables in the same way, we arrive at the claimed normalization. This concludes the proof.
\end{proof}
\begin{lem}\label{lemInCond}
For fixed $\vec{\zeta}\in W^N$, the transition density $r_{t}(\vec{\zeta},\vec{\xi})$ as given by \eqref{eqrtau}, satisfies
\begin{equation}
	\lim_{{t}\to0}\int_{W^N}r_{t}(\vec{\zeta},\vec{\xi})f(\xi)\,\D \vec{\xi}=f(\vec{\zeta})
\end{equation}
for any $C^\infty$ function \mbox{$f\colon W^N\rightarrow\R$}, whose support is compact and has a distance of at least some $\e>0$ to the boundary of $W^N$.
\end{lem}
\begin{proof}
At first consider the contribution to the determinant in \eqref{eqrtau} coming from the diagonal. For $k=l$ the products in \eqref{eqFkl} cancel out, so we are left with a simple gaussian density. This contribution is thus given by the multidimensional heat kernel, which is well known to converge to the delta distribution. The remaining task is to prove that for all other permutations the integral vanishes in the limit.

Let $\sigma$ be such a permutation. Its contribution is
\begin{equation}\label{eqL2.2a}
	\int_{\R^N}\D\vec{\xi}\,f(\vec{\xi})\prod_{k=1}^NF_{k,\sigma(k)}(\xi_{N+1-\sigma(k)}-\zeta_{N+1-k},{t}),
\end{equation}
where we have extended the domain of $f$ to $\R^N$, being identically zero outside of $W_N$. We also omitted the prefactor since it is bounded for $\xi$ in the compact domain of $f$.

There exist $i<j$ with $\sigma(j)\leq i<\sigma(i)$. Let
\begin{equation}\begin{aligned}
	\widetilde{W}_1&=\{\vec{\xi}\in \R^N\colon \xi_{N+1-\sigma(i)}-\zeta_{N+1-i}<-\e/2\}\\
	\widetilde{W}_2&=\{\vec{\xi}\in \R^N\colon \xi_{N+1-\sigma(j)}-\zeta_{N+1-j}>\e/2\}.
\end{aligned}\end{equation}
It is enough to restrict the area of integration to these two sets, since on the complement of $\widetilde{W}_1\cup\widetilde{W}_2$, we have
\begin{equation}
	\xi_{N+1-\sigma(i)}\geq\zeta_{N+1-i}-\e/2\geq\zeta_{N+1-j}-\e/2\geq\xi_{N+1-\sigma(j)}-\e,
\end{equation}
so we are not inside the support of $f$.

We start with the contribution coming from $\widetilde{W}_1$. Notice that by
\begin{equation}
	F_{k,l}(\xi,{t})=e^{-\xi\mu_{N+1-l}}\frac{\D}{\D\xi}\Big(e^{\xi\mu_{N+1-l}}F_{k,l+1}(\xi,{t})\Big),
\end{equation}
all functions $F_{k,l}$ with $k>l$ can be written as iterated derivatives of $F_{k,k}$ and some exponential functions. For each $k\neq i$ with $k>\sigma(k)$ we write $F_{k,\sigma(k)}$ in this way and then use partial integration to move the exponential factors and derivatives onto $f$. The result is
\begin{equation}\label{eq2.20}
	\int_{\widetilde{W}_1}\D\vec{\xi}\,\tilde{f}(\vec{\xi})F_{i,\sigma(i)}(\xi_{N+1-\sigma(i)}-\zeta_{N+1-i},{t})\prod_{k\neq i}F_{k,\max\{k,\sigma(k)\}}(\xi_{N+1-\sigma(k)}-\zeta_{N+1-k},{t})
\end{equation}
for a new $C^\infty$ function $\tilde{f}$, which has compact support and is therefore bounded, too. We can bound the contribution by first integrating the variables $\xi_{N+1-\sigma(k)}$ with $k\geq\sigma(k)$, $k\neq i$, where we have a gaussian factor $F_{k,k}$:
\begin{equation}\begin{aligned}
	\left|\eqref{eq2.20}\right|\leq\sup_{\vec{x}}\big|\tilde{f}(\vec{x})\big|\int_{\widetilde{W}_1'}&\left|F_{i,\sigma(i)}(\xi_{N+1-\sigma(i)}-\zeta_{N+1-i},{t})\right|\D\xi_{N+1-\sigma(i)}\,\\&\prod_{k< \sigma(k),k\neq i}\left|F_{k,\sigma(k)}(\xi_{N+1-\sigma(k)}-\zeta_{N+1-k},{t})\right|\D\xi_{N+1-\sigma(k)}.
\end{aligned}\end{equation}
$\widetilde{W}_1'$ consists of the yet to be integrated $\xi$-components that are contained in the set \mbox{$\widetilde{W}_1\cap\supp(\widetilde{f})$}. In particular, $\widetilde{W}_1'$ is compact, so the functions $F_{k,\sigma(k)}$, $k\neq i$, are bounded uniformly in ${t}$ by Lemma~\ref{lemFkl}. The remaining integral gives:
\begin{equation}
	\left|\eqref{eq2.20}\right|\leq\const\int^{-\e/2}_{-\infty}\left|F_{i,\sigma(i)}(x,{t})\right|\D x,
\end{equation}
which converges to $0$ as ${t}\to0$ by \eqref{eqFkllimn}.

The contribution of $\widetilde{W}_2$ can be bounded analogously with $j$ playing the role of $i$. The final convergence is then given by \eqref{eqFkllimp}.
\end{proof}

\begin{lem}\label{lemFkl}
For each $\e>0$ we have
\begin{align}
	\lim_{{t}\to0}\int_\e^\infty\left|F_{k,l}(x,{t})\right|\D x&=0, & 1&\leq l\leq k\leq N,\label{eqFkllimp}\\
	\lim_{{t}\to0}\int^{-\e}_{-\infty}\left|F_{k,l}(x,{t})\right|\D x&=0, & 1&\leq k,l\leq N.\label{eqFkllimn}
\end{align}
In addition, for each $1\leq k<l\leq N$ the function $F_{k,l}(x,{t})$ is bounded uniformly in ${t}$ on compact sets.
\end{lem}

\begin{proof}
Let $x<-\e$, and choose a $\mu$ which is positive. By a transformation of variable we have
\begin{equation}\begin{aligned}\label{eqFklbound}
	\left|F_{k,l}(x,{t})\right|&=\bigg|\frac{1}{2\pi\I}\int_{\I\R+\mu}\D w\,e^{{t} w^2/2+x w}\frac{\prod_{i=1}^{k-1}(w+\mu_{N+1-i})}{\prod_{i=1}^{l-1}(w+\mu_{N+1-i})}\bigg|\\
	&=\bigg|\frac{1}{2\pi\I}\int_{\I\R+\mu}\D v\,\sqrt{{t}}^{l-k-1}e^{ v^2/2+x v/\sqrt{{t}}}\frac{\prod_{i=1}^{k-1}(v+\sqrt{{t}}\mu_{N+1-i})}{\prod_{i=1}^{l-1}(v+\sqrt{{t}}\mu_{N+1-i})}\bigg|\\
	&\leq(2\pi)^{-1}\sqrt{{t}}^{l-k-1}e^{x\mu/\sqrt{{t}}}\int_{\I\R+\mu}|\D v|\,e^{ \Re(v^2/2)}g(|v|),
\end{aligned}\end{equation}
where $g(|v|)$ denotes a bound on the fraction part of the integrand, which grows at most polynomial in $|v|$. Convergence of the integral is ensured by the exponential term, so integrating and taking the limit ${t}\to0$ gives \eqref{eqFkllimn}. To see \eqref{eqFkllimp}, notice that by $l\leq k$ the integrand has no poles, so we can shift the contour to the right, such that $\mu$ is negative, and obtain the convergence analogously.

We are left to prove uniform boundedness of $F_{k,l}$ on compact sets for $k<l$. For $x\leq0$ we can use \eqref{eqFklbound} to get
\begin{equation}
	\left|F_{k,l}(x,{t})\right|\leq(2\pi)^{-1} \int_{\I\R+\mu}|\D v|\,e^{ \Re(v^2/2)}g(|v|)
\end{equation}
for ${t}\leq1$. In the case $x>0$ we shift the contour to negative $\mu$, thus obtaining contributions from residua as well as from the remaining integral. The latter can be bounded as before, while the residua are well-behaved functions, which converge uniformly on compact sets.
\end{proof}

\subsection{Transition density for random initial positions}
To obtain a representation as a signed determinantal point process we have to introduce a new measure. This measure $\Pb_+$ coincides with $\Pb$ on the sigma algebra which is generated by $\zeta_{k+1}-\zeta_{k}$, $k\in\Z$, and the driving Brownian motions $B_k$, $k\in\Z$. But under $\Pb_+$, $\zeta_0$ is a random variable with an exponential distribution instead of being fixed at zero. Formally, \mbox{$\Pb_+=\Pb\otimes\Pb_{\zeta_0}$}, with $\Pb_{\zeta_0}$ giving rise to \mbox{$\zeta_0\sim\exp(\lambda-\rho)$}, so that $\Pb$ is the result of conditioning $\Pb_+$ on the event $\{\zeta_0=0\}$. This new measure satisfies a determinantal formula for the joint distribution at a fixed time.

\begin{prop}\label{PropRandomIC}
Under the modified initial condition specified by $\Pb_+$, the joint density of the positions of the reflected Brownian motions \mbox{$\{x_n({t}),0\leq n\leq N-1\}$} is given by
	\begin{equation}\begin{aligned}
	\Pb_+(\vec{x}({t})\in \D\vec{\xi})=(\lambda-\rho)\lambda^{N-1}e^{-{t}\rho^2/2+\rho \xi_0}
	 \det_{1\leq k,l\leq N}[\widetilde{F}_{k-l}(\xi_{N-l},{t})]\,\D\vec{\xi}
\end{aligned}\label{Prnd}\end{equation}
with
\begin{equation}
	\widetilde{F}_{k}(\xi,{t}):=\frac{1}{2\pi\I}\int_{\I\R+\e}dw\,\frac{e^{{t} w^2/2+\xi w}w^k}{w+\lambda}.
\end{equation}
\end{prop}

For the related model, the totally asymmetric simple exclusion process, a formula similar to the one of Proposition~\ref{PropRandomIC} also exists~\cite{Bor08_privatecomm}. Here we provide a direct proof of it.

\begin{proof}[Proof of Proposition~\ref{PropRandomIC}]
The fixed time distribution can be obtained by integrating the transition density \eqref{transDens} over the initial condition. Denote by $p_+(\vec{\xi})$ the probability density of $\vec{x}(t)$, i.e., $\Pb_+(\vec{x}({t})\in \D\vec{\xi})=p_+(\vec{\xi})\D\vec{\xi}$. $p_+(\vec{x})$ equals
\begin{equation}\label{eqPrnd1}\begin{aligned}
	&\int_{W^N\cap\{\zeta_0>0\}}\hspace{-1cm}\D\vec{\zeta}\, e^{\rho(\xi_0-\zeta_0)-{t}\rho^2/2}(\lambda-\rho)\lambda^{N-1}e^{\rho\zeta_0}e^{-\lambda\zeta_{N-1}}
\det_{1\leq k,l\leq N}[F_{k,l}(\xi_{N-l}-\zeta_{N-k},{t})]\\
=&(\lambda-\rho)\lambda^{N-1}e^{-{t}\rho^2/2+\rho\xi_0}\int_{W^N\cap\{\zeta_0>0\}}\D\vec{\zeta}\, e^{\lambda\zeta_{N-1}} \\
	&\quad\times\det_{1\leq k,l\leq N}\left[\frac{1}{2\pi\I}\int_{\I\R+\mu}\D w_k\,e^{{t} w_k^2/2}e^{\xi_{N-l}w_k}e^{-\zeta_{N-k}w_k}w_k^{k-l}\right]\\
	=&(\lambda-\rho)\lambda^{N-1}e^{-{t}\rho^2/2+\rho\xi_0}\sum_{\sigma\in S_N}(-1)^{|\sigma|}\prod_{k=1}^N\int_{\I\R+\mu}\frac{\D w_k}{2\pi\I}\,e^{{t} w_k^2/2}e^{\xi_{N-\sigma(k)}w_k}w_k^{k-\sigma(k)}\\ &\quad\times\int^{\infty}_0d\zeta_0\dots\int^{\infty}_{\zeta_{N-2}}d\zeta_{N-1}\,e^{-\lambda\zeta_{N-1}}e^{-\zeta_{N-1}w_1}e^{-\zeta_{N-2}w_2}\dots e^{-\zeta_0w_N}\\
=&(\lambda-\rho)\lambda^{N-1}e^{-{t}\rho^2/2+\rho\xi_0}\sum_{\sigma\in S_N}(-1)^{|\sigma|}\prod_{k=1}^N\int_{\I\R+\mu}\frac{\D w_k}{2\pi\I}\,\frac{e^{{t} w_k^2/2}e^{\xi_{N-\sigma(k)}w_k}w_k^{k-\sigma(k)}}{w_1+\dots+w_k+\lambda}.
\end{aligned}\end{equation}
Since all $w_k$ are integrated over the same contour, we can replace $w_k$ by $w_{\sigma(k)}$:
\begin{equation}\label{eqPrnd2}\begin{aligned}
\eqref{eqPrnd1}&=(\lambda-\rho)\lambda^{N-1}e^{-{t}\rho^2/2+\rho\xi_0}\\
	&\quad\times\sum_{\sigma\in S_N}(-1)^{|\sigma|}\prod_{k=1}^N\int_{\I\R+\mu}\frac{\D w_k}{2\pi\I}\,\frac{e^{{t} w_k^2/2}e^{\xi_{N-\sigma(k)}w_{\sigma(k)}}w_{\sigma(k)}^{k-\sigma(k)}}{w_{\sigma(1)}+\dots+w_{\sigma(k)}+\lambda}\\ &=(\lambda-\rho)\lambda^{N-1}e^{-{t}\rho^2/2+\rho\xi_0}\prod_{k=1}^N\int_{\I\R+\mu}\frac{\D w_k}{2\pi\I}\,e^{{t} w_k^2/2}e^{\xi_{N-k}w_k}w_k^{-k}\\
&\quad\times\sum_{\sigma\in S_N}(-1)^{|\sigma|}\prod_{k=1}^N\frac{w_{\sigma(k)}^{k}}{w_{\sigma(1)}+\dots+w_{\sigma(k)}+\lambda}.
\end{aligned}\end{equation}
We apply Lemma~\ref{lemDetIdent} below to the sum and finally obtain
\begin{equation}\begin{aligned}
p_+(\vec{x})&=(\lambda-\rho)\lambda^{N-1}e^{-{t}\rho^2/2+\rho\xi_0}\prod_{k=1}^N\int_{\I\R+\mu}\frac{\D w_k}{2\pi\I}\,e^{{t} w_l^2/2}e^{\xi_{N-l}w_l}w_l^{-l} \det_{1\leq k,l\leq N}\left[\frac{w_l^k}{w_l+\lambda}\right]\\
&=(\lambda-\rho)\lambda^{N-1}e^{-{t}\rho^2/2+\rho\xi_0}\det_{1\leq k,l\leq N}\left[\widetilde{F}_{k-l}(\xi_{N-l},{t})\right].
\end{aligned}\end{equation}
\end{proof}

\begin{lem}\label{lemDetIdent}
Given $N\in\N$, $\lambda>0$ and $w_1,\dots,w_N\in\C\setminus\R_-$, the following identity holds:
\begin{equation}\label{eqDetIdent}
	\sum_{\sigma\in S_N}(-1)^{|\sigma|}\prod_{k=1}^N\frac{w_{\sigma(k)}^{k}}{w_{\sigma(1)}+\dots+w_{\sigma(k)}+\lambda}=\det_{1\leq k,l\leq N}\left[\frac{w_l^k}{w_l+\lambda}\right].
\end{equation}
\end{lem}

\begin{proof}
We use induction on $N$. For $N=1$ the statement is trivial. For arbitrary $N$, rearrange the left hand side of \eqref{eqDetIdent} as
\begin{equation}\begin{aligned}\label{eqDI1}
	&\sum_{\sigma\in S_N}(-1)^{|\sigma|}\prod_{k=1}^N\frac{w_{\sigma(k)}^{k}}{w_{\sigma(1)}+\dots+w_{\sigma(k)}+\lambda}\\&\quad=\sum_{l=1}^N\frac{w_l^N}{w_1+\dots+w_N+\lambda}\sum_{\sigma\in S_{N},\sigma(N)=l}(-1)^{|\sigma|}\prod_{k=1}^{N-1}\frac{w_{\sigma(k)}^{k}}{w_{\sigma(1)}+\dots+w_{\sigma(k)}+\lambda}\\
	&\quad=\sum_{l=1}^N\frac{w_l^N}{w_1+\dots+w_N+\lambda}\sum_{\sigma\in S_{N},\sigma(N)=l}(-1)^{|\sigma|}\prod_{k=1}^{N-1}\frac{w_{\sigma(k)}^{k}}{w_{\sigma(k)}+\lambda},
\end{aligned}\end{equation}
where we applied the induction hypothesis to the second sum. Further,
\begin{equation}\begin{aligned}\label{eqDI2}
	\eqref{eqDI1}=&\sum_{\sigma\in S_N}(-1)^{|\sigma|}\frac{w_{\sigma(N)}+\lambda}{w_1+\dots+w_N+\lambda}\prod_{k=1}^{N}\frac{w_{\sigma(k)}^{k}}{w_{\sigma(k)}+\lambda}\\
	&=\frac{1}{w_1+\dots+w_N+\lambda}\prod_{l=1}^N\frac{w_l}{w_l+\lambda}\\
&\quad\times\bigg(\sum_{\sigma\in S_N}(-1)^{|\sigma|}w_{\sigma(N)}\prod_{k=1}^Nw_{\sigma(k)}^{k-1}+\lambda\sum_{\sigma\in S_N}(-1)^{|\sigma|}\prod_{k=1}^Nw_{\sigma(k)}^{k-1}\bigg)\\
	&=\frac{1}{w_1+\dots+w_N+\lambda}\prod_{l=1}^N\frac{w_l}{w_l+\lambda}\\
&\quad \times\left(\det_{1\leq k,l\leq N}\left[w_l^{k-1+\delta_{k,N}}\right]+\lambda\det_{1\leq k,l\leq N}\left[w_l^{k-1}\right]\right),
\end{aligned}\end{equation}
with $\delta_{k,N}$ being the Kronecker delta. Inserting the identity
\begin{equation}\label{eqMuir}
	\det_{1\leq k,l\leq N}\left[w_l^{k-1+\delta_{k,N}}\right]=(w_1+\dots+w_N)\det_{1\leq k,l\leq N}\left[w_l^{k-1}\right],
\end{equation}
we arrive at
\begin{equation}
\eqref{eqDI1}=\bigg(\prod_{l=1}^N\frac{w_l}{w_l+\lambda}\bigg)\det_{1\leq k,l\leq N}\left[w_l^{k-1}\right]=\det_{1\leq k,l\leq N}\bigg[\frac{w_l^k}{w_l+\lambda}\bigg].
\end{equation}

To show \eqref{eqMuir} we introduce the variable $w_{N+1}$ and consider the factorization
\begin{equation}
	\det_{1\leq k,l\leq N+1}\left[w_l^{k-1}\right]=\prod_{i=1}^N(w_{N+1}-w_i)\det_{1\leq k,l\leq N}\left[w_l^{k-1}\right],
\end{equation}
which follows directly from the explicit formula for a Vandermonde determinant. Expanding the determinant on the left hand side along the $(N+1)$-th column gives an explicit expression in terms of monomials in $w_{N+1}$. Examining the coefficient of $w_{N+1}^{N-1}$ on the left and right hand side respectively provides \eqref{eqMuir}.
\end{proof}

\subsection{Proof of Proposition~\ref{propKernel}}
We can rewrite the measure in Proposition~\ref{PropRandomIC} in terms of a conditional \mbox{$L$-ensemble} (see Lemma 3.4 of~\cite{BFPS06} reported here as Lemma~\ref{lemDetMeasure}) and obtain a Fredholm determinant expression for the joint distribution of any subsets of particles position. Then it remains to relate the law under $\Pb_+$ and $\Pb$, which is the law of the reflected Brownian motions specified by the initial condition \eqref{statModel}. This is made using a \emph{shift argument}, analogue to the one used for the polynuclear growth model with external sources~\cite{BR00,SI04} or in the totally asymmetric simple exclusion process~\cite{PS02b,FS05a,BFP09}.

\begin{proof}[Proof of Proposition~\ref{propKernel}]
The proof is divided into two steps. In \emph{Step 1} we determine the distribution under $\Pb_+$ and in \emph{Step 2} we extend this result via a shift argument to $\Pb$.

\emph{Step 1.} We consider the law of the process under $\Pb_+$ for now. The first part of the proof is identical to the proof of Proposition 3.5~\cite{FSW13}, so it is only sketched here.
Using repeatedly the identity
\begin{equation}
	\widetilde{F}_{k}(\xi,{t})=\int^\xi_{-\infty} dx\,\widetilde{F}_{k+1}(x,{t}),
\end{equation}
relabeling $\xi_1^k:=\xi_{k-1}$, and introducing new variables $\xi_l^k$ for $2\leq l\leq k\leq N$,
we can write
\begin{equation}
	\det_{1\leq k,l\leq N}\big[\widetilde{F}_{k-l}(\xi_1^{N+1-l},{t})\big]=\int_{\mathcal D'} \det_{1\leq k,l\leq N}\big[\widetilde{F}_{k-1}(\xi_l^{N},{t})\big]\prod_{2\leq l\leq k\leq N} \D\xi_l^k,
\end{equation}
where $\mathcal D' = \{\xi_l^k\in\R,2\leq l\leq k\leq N|x_l^k\leq x_{l-1}^{k-1}\}$. Using the antisymmetry of the determinant and encoding the constraint on the integration variables into indicator functions, we obtain that the measure \eqref{Prnd} is a marginal of
\begin{equation}\label{ExtM}\begin{aligned}
	\const&\cdot e^{\rho\xi^1_1}\prod_{n=2}^{N}\det_{1\leq i,j\leq n}\left[\Id_{\xi_i^{n-1}\leq\xi_j^n}\right]\det_{1\leq k,l\leq N}\big[ \widetilde{F}_{k-1}(\xi_l^{N},{t})\big]\\=
	\const&\cdot \prod_{n=1}^{N}\det_{1\leq i,j\leq n}\big[\tilde{\phi}_n(\xi_i^{n-1},\xi_j^n)\big]\det_{1\leq k,l\leq N}\big[ \widetilde{F}_{k-1}(\xi_l^{N},{t})\big]
\end{aligned}\end{equation}
with
\begin{equation}\begin{aligned}
	\tilde{\phi}_n(x,y)&=\Id_{x\leq y}, \quad \text{for } n\geq2\\
	\tilde{\phi}_1(x,y)&=e^{\rho y},
\end{aligned}\end{equation}
and using the convention that $\xi_n^{n-1}\leq y$ always holds.

The measure \eqref{ExtM} has the appropiate form for applying Lemma~\ref{lemDetMeasure}. The composition of the $\tilde{\phi}$ functions can be evaluated explicitly as
\begin{equation}\begin{aligned}
	\tilde{\phi}_{0,n}(x,y)&=(\tilde{\phi}_1*\dots*\tilde{\phi}_{n})(x,y)=\rho^{1-n}e^{\rho y}, & &\text{for } n\geq1,\\
	\tilde{\phi}_{m,n}(x,y)&=(\tilde{\phi}_{m+1}*\dots*\tilde{\phi}_{n})(x,y)=\frac{(y-x)^{n-m-1}}{(n-m-1)!}\Id_{x\leq y}, & &\text{for } n>m\geq1.
\end{aligned}\end{equation}
Define
\begin{equation}
	\Psi^n_{n-k}(\xi):=\frac{(-1)^{n-k}}{2\pi\I}\int_{\I\R-\e}\D w\,\,\frac{e^{{t} w^2/2+\xi w}w^{n-k}}{w+\lambda},
\end{equation}
for $n,k\geq1$ and some $0<\e<\lambda$. In the case $n\geq k$ the integrand has no poles in the region $|w|<\lambda$, which implies $\Psi^n_{n-k}=(-1)^{n-k}\widetilde{F}_{n-k}$. The straightforward recursion
\begin{equation}
	(\tilde{\phi}_n*\Psi^n_{n-k})(\xi)=\Psi^{n-1}_{n-1-k}(\xi)
\end{equation}
eventually leads to condition \eqref{Sasdef_psi} being satisfied.

The space $V_n$ is generated by
\begin{equation}
	\{\tilde{\phi}_{0,n}(\xi_1^0,x),\dots,\tilde{\phi}_{n-2,n}(\xi_{n-1}^{n-2},x),\tilde{\phi}_{n-1,n}(\xi_n^{n-1},x)\},
\end{equation}
so a basis for $V_n$ is given by
\begin{equation}
	\{e^{\rho x},x^{n-2},x^{n-3},\dots,x,1\}.
\end{equation}
Choose functions $\Phi^n_{n-k}$ as follows
\begin{equation}
		\Phi^n_{n-k}(\xi)=\begin{cases}
\frac{(-1)^{n-k}}{2\pi\I}\oint_{\Gamma_0}\D z\,\frac{z+\lambda}{e^{{t} z^2/2+\xi z}z^{n-k+1}}& 2\leq k\leq n,\\
\frac{(-1)^{n-1}}{2\pi\I}\oint_{\Gamma_{0,-\rho}}\D z\,\frac{z+\lambda}{e^{{t} z^2/2+\xi z}z^{n-1}(z+\rho)}& k=1.
\end{cases}
\end{equation}
By residue calculating rules, $\Phi^n_{n-k}$ is a polynomial of order $n-k$ for $k\geq2$ and a linear combination of $1$ and $e^{\rho\xi}$
for $k=1$, so these functions indeed generate $V_n$. To show \eqref{Sasortho} for $\ell\geq2$, we decompose the scalar product as follows:
\begin{equation}\label{eq2.41}
\int_{\R_-} \D \xi\, \Psi^{n}_{n-k}(\xi) \Phi^n_{n-\ell}(\xi) + \int_{\R_+} \D \xi\, \Psi^{n}_{n-k}(\xi) \Phi^n_{n-\ell}(\xi).
\end{equation}
Since $n-k\geq0$ we are free to choose the sign of $\e$ as necessary. For the first term, we choose $\e<0$ and the path $\Gamma_0$ close enough to zero, such that always \mbox{$\Re(w-z)>0$}. Then, we can take the integral over $\xi$ inside and obtain
\begin{equation}
\int_{\R_-} \D \xi\, \Psi^{n}_{n-k}(\xi) \Phi^n_{n-\ell}(\xi) =\frac{(-1)^{k-l}}{(2\pi\I)^2}\int_{\I\R-\e}\D w \oint_{\Gamma_0}\D z\, \frac{e^{{t} w^2/2} w^{n-k}(z+\lambda)}{e^{{t} z^2/2}z^{n-\ell+1}(w+\lambda)(w-z)}.
\end{equation}
For the second term, we choose $\e>0$ to obtain \mbox{$\Re(w-z)<0$}. Then again, we can take the integral over $\xi$ inside and arrive at the same expression up to a minus sign. The net result of \eqref{eq2.41} is a residue at $w=z$, which is given by
\begin{equation}
\frac{(-1)^{k-l}}{2\pi\I}\oint_{\Gamma_0}\D z\, z^{\ell-k-1}=\delta_{k,\ell}.
\end{equation}
The case $\ell=1$ uses the same decomposition and requires the choice $\e>\rho$ resp.\ $\e<0$, finally leading to
\begin{equation}
\eqref{eq2.41}=\frac{(-1)^{k-1}}{2\pi\I}\oint_{\Gamma_{0,-\rho}}\D z\, \frac{z^{1-k}}{z+\rho}=\delta_{k,1}.
\end{equation}

Furthermore, both $\tilde{\phi}_n(\xi_{n}^{n-1},x)$ and $\Phi_0^{n}(\xi)$ are constants, so the kernel has a simple form (compare with \eqref{SasK})
\begin{equation}
\tilde{\mathcal{K}}(n_1,\xi_1;n_2,\xi_2)=-\tilde{\phi}_{n_1,n_2}(\xi_1,\xi_2)\Id_{(n_2>n_1)} + \sum_{k=1}^{n_2} \Psi_{n_1-k}^{n_1}(\xi_1) \Phi_{n_2-k}^{n_2}(\xi_2).
\end{equation}
However, the relabeling $\xi_1^k:=\xi_{k-1}$ included a index shift, so the kernel of our system is actually
\begin{equation}\begin{aligned}
\mathcal{K}(n_1,\xi_1;n_2,\xi_2)&=\tilde{\mathcal{K}}(n_1+1,\xi_1;n_2+1,\xi_2)\\&=-\phi_{n_1,n_2}(\xi_1,\xi_2)\Id_{(n_2>n_1)} + \sum_{k=1}^{n_2} \Psi_{n_1-k+1}^{n_1+1}(\xi_1) \Phi_{n_2-k+1}^{n_2+1}(\xi_2).
\end{aligned}\end{equation}

Note that we are free to extend the summation over $k$ up to infinity, since the integral expression for $\Phi_{n-k}^{n}(\xi)$ vanishes for $k>n$ anyway. Taking the sum inside the integrals we can write
\begin{equation}\label{eqKernSum}\begin{aligned}
	\sum_{k\geq1} \Psi_{n_1-k+1}^{n_1+1}(\xi_1) \Phi_{n_2-k+1}^{n_2+1}(\xi_2)=\frac{1}{(2\pi\I)^2}\int_{\I\R-\e}\hspace{-0.5cm}\D w \oint_{\Gamma_{0,-\rho}}\hspace{-0.5cm}\D z\,	\frac{e^{{t} w^2/2+\xi_1 w}(-w)^{n_1}}{e^{{t} z^2/2+\xi_2 z}(-z)^{n_2}}\eta(w,z),
\end{aligned}\end{equation}
with
\begin{equation}
	\eta(w,z)=\frac{z+\lambda}{(w+\lambda)(z+\rho)}+\sum_{k\geq2}\frac{z^{k-2}(z+\lambda)}{w^{k-1}(w+\lambda)}.
\end{equation}
By choosing contours such that $|z|<|w|$, we can use the formula for a geometric series, resulting in
\begin{equation}\begin{aligned}
	\eta(w,z)&=\frac{z+\lambda}{(w+\lambda)(z+\rho)}+\frac{z+\lambda}{(w+\lambda)w}\frac{1}{1-z/w}\\
	&=\frac{1}{w-z}+\frac{\lambda-\rho}{(w+\lambda)(z+\rho)}.
\end{aligned}\end{equation}
Inserting this expression back into \eqref{eqKernSum} gives the kernel \eqref{eqKtflat}, which governs the multidimensional distributions of $x_n({t})$ under the measure $\Pb_+$, namely
\begin{equation}\label{eqDistP+}
\Pb_+\bigg(\bigcap_{n\in S} \{x_n({t})\leq a_n\}\bigg)=\det(\Id-\chi_a \mathcal{K} \chi_a)_{L^2(S\times\R)}.
\end{equation}

\emph{Step 2.} The distributions under $\Pb$ and under $\Pb_+$ can be related via the following \emph{shift argument}. Introducing the shorthand
\begin{equation}
	\widetilde{\mathcal{E}}(S,\vec{a}):=\bigcap_{n\in S} \{x_n({t})\leq a_n\},
\end{equation}
we have
\begin{equation}\begin{aligned}
		\Pb_+(\widetilde{\mathcal{E}}(S,\vec{a}))&=\int_{\R_+} \D x\, \Pb_+(x_0(0)\in \D x)\Pb_+(\widetilde{\mathcal{E}}(S,\vec{a})|x_0(0)=x)\\
		&=\int_{\R_+} \D x\,(\lambda-\rho)e^{-(\lambda-\rho)x}\Pb(\widetilde{\mathcal{E}}(S,\vec{a}-x))\\
		&=-e^{-(\lambda-\rho)x}\,\Pb(\widetilde{\mathcal{E}}(S,\vec{a}-x))\Big|_0^\infty+\int_{\R_+} \D x\,e^{-(\lambda-\rho)x}\frac{\D}{\D x}\Pb(\widetilde{\mathcal{E}}(S,\vec{a}-x))\\
		&=\Pb(\widetilde{\mathcal{E}}(S,\vec{a}))-\int_{\R_+} \D x\,e^{-(\lambda-\rho)x}\sum_{k\in S}\frac{\D}{\D a_k}\Pb(\widetilde{\mathcal{E}}(S,\vec{a}))\\
		&=\Pb(\widetilde{\mathcal{E}}(S,\vec{a}))-\frac{1}{\lambda-\rho}\sum_{k\in S}\frac{\D}{\D a_k}\Pb_+(\widetilde{\mathcal{E}}(S,\vec{a})).
\end{aligned}\end{equation}
Combining the identity
\begin{equation}
	\Pb(\widetilde{\mathcal{E}}(S,\vec{a}))=\bigg(1+\frac{1}{\lambda-\rho}\sum_{k\in S}\frac{\D}{\D a_k}\bigg)\Pb_+(\widetilde{\mathcal{E}}(S,\vec{a}))
\end{equation}
with \eqref{eqDistP+} finishes the proof.
\end{proof}	

\begin{lem}[Corollary of Theorem~4.2~\cite{BF07}]\label{lemDetMeasure}
Assume we have a signed measure on $\{x_i^n,n=1,\dotsc,N,i=1,\dotsc,n\}$ given in the form,
\begin{equation}\label{Sasweight}
 \frac{1}{Z_N}\prod_{n=1}^{N} \det[\phi_n(x_i^{n-1},x_j^n)]_{1\leq i,j\leq n} \det[\Psi_{N-i}^{N}(x_{j}^N)]_{1\leq i,j \leq N},
\end{equation}
where $x_{n+1}^n$ are some ``virtual'' variables and $Z_N$ is a normalization constant. If $Z_N\neq 0$, then the correlation functions are determinantal.

To write down the kernel we need to introduce some notations. Define
\begin{equation}\label{Sasdef phi12}
\phi^{(n_1,n_2)}(x,y)=
\begin{cases} (\phi_{n_1+1} \ast \dotsb \ast \phi_{n_2})(x,y),& n_1<n_2,\\
0,& n_1\geq n_2,
\end{cases}.
\end{equation}
where $(a* b)(x,y)=\int_\R \D z\, a(x,z) b(z,y)$, and, for $1\leq n<N$,
\begin{equation}\label{Sasdef_psi}
\Psi_{n-j}^{n}(x) := (\phi^{(n,N)} * \Psi_{N-j}^{N})(y), \quad j=1,\dotsc,N.
\end{equation}
Then the functions
\begin{equation}
\{ \phi^{(0,n)}(x_1^0,x), \dots,\phi^{(n-2,n)}(x_{n-1}^{n-2},x), \phi_{n}(x_{n}^{n-1},x)\}
\end{equation}
are linearly independent and generate the $n$-dimensional space $V_n$. Define a set of functions $\{\Phi_{n-j}^{n}(x), j=1,\dotsc,n\}$ spanning $V_n$ defined by the orthogonality relations
\begin{equation}\label{Sasortho}
\int_\R \D x\, \Phi_{n-i}^n(x) \Psi_{n-j}^n(x) = \delta_{i,j}
\end{equation}
for $1\leq i,j\leq n$.

Further, if $\phi_n(x_n^{n-1},x)=c_n \Phi_0^{n}(x)$, for some $c_n\neq 0$, \mbox{$n=1,\dotsc,N$}, then the kernel takes the simple form
\begin{equation}\label{SasK}
K(n_1,x_1;n_2,x_2)= -\phi^{(n_1,n_2)}(x_1,x_2)+ \sum_{k=1}^{n_2} \Psi_{n_1-k}^{n_1}(x_1) \Phi_{n_2-k}^{n_2}(x_2).
\end{equation}
\end{lem}

\section{Asymptotic analysis - Proof of Theorem~\ref{thmAsymp}}\label{SectAsymtotics}
\begin{remark}\label{remLambda}
The change in variables
\begin{equation}\begin{aligned}
	x&\mapsto\lambda^{-1}x&\quad {t}&\mapsto\lambda^{-2}{t}
\end{aligned}\end{equation}
reproduces the same system with new parameters $\widetilde{\lambda}=1$ and $\widetilde{\rho}=\frac{\rho}{\lambda}$. We can therefore restrict our considerations to $\lambda=1$ without loss of generality.
\end{remark}
Fix $\lambda=1$ from now on. According to \eqref{eqScaledProcess} we use the scaled variables
\begin{equation}\label{eqScaling}\begin{aligned}
	n_i&={t}+2{t}^{2/3}r_i\\
	\xi_i&=2{t}+2{t}^{2/3}r_i+{t}^{1/3}s_i\\
	\rho&=1-{t}^{-1/3}\delta,
\end{aligned}\end{equation}
with $\delta>0$. Correspondingly, consider the rescaled (and conjugated) kernel
\begin{equation}
	\mathcal{K}^{\rm resc}(r_1,s_1;r_2,s_2)={t}^{1/3}e^{\xi_1-\xi_2}\mathcal{K}(n_1,\xi_1;n_2,\xi_2),
\end{equation}
which naturally decomposes into
\begin{equation}
	\mathcal{K}^{\rm resc}(r_1,s_1;r_2,s_2)=-\phi_{r_1,r_2}^{\rm resc}(s_1,s_2)\Id_{(r_1<r_2)}+\mathcal{K}_0^{\rm resc}(r_1,s_1;r_2,s_2).
\end{equation}

\begin{rem}
Instead of integrals over Airy functions \eqref{eqKernelDef} can also be written as contour integrals:
\begin{equation}\label{contInt}\begin{aligned}
	K_{r_1,r_2}(s_1,s_2)&=\frac{-1}{(2\pi\I)^2}\int_{e^{-2\pi\I/3}\infty}^{e^{2\pi\I/3}\infty} \D W\int_{e^{\pi\I/3}\infty}^{e^{-\pi\I/3}\infty} \D Z\frac{e^{Z^3/3+r_2Z^2-s_2Z}}{e^{W^3/3+r_1W^2-s_1W}}\frac{1}{Z-W}\\
		f_{r_1}(s_1)&=\frac{1}{2\pi\I}\int_{e^{-2\pi\I/3}\infty,\text{ right of }0}^{e^{2\pi\I/3}\infty}   \D W\frac{e^{-(W^3/3+r_1W^2-s_1W)}}{W}\\
		g_{r_2}(s_2)&=\frac{1}{2\pi\I}\int_{e^{\pi\I/3}\infty,\text{ left of }\delta}^{e^{-\pi\I/3}\infty} \D Z\frac{e^{Z^3/3+r_2Z^2-s_2Z}}{Z-\delta}.
\end{aligned}\end{equation}
In the integral defining $K$, the path for $W$ and $Z$ do not have to intersect. In addition, the Gaussian part has a representation in terms of an integral over Airy functions:
\begin{equation}\label{eqVRep}
 V_{r_1,r_2}(s_1,s_2)=\frac{e^{\frac{2}{3}r_2^3+r_2s_2}}{e^{\frac{2}{3}r_1^3+r_1s_1}}\int_\R\D x\, e^{-x(r_1-r_2)}\Ai(r_1^2+s_1+x)\Ai(r_2^2+s_2+x).
\end{equation}

\end{rem}

In order to establish the asymptotics of the joint distributions, one needs both a pointwise limit of the kernel, as well as uniform bounds to ensure convergence of the Fredholm determinant expansion. The first time this approach was used is in~\cite{GTW00}. These results are contained in the following propositions.
\begin{prop}\label{propPointw}
Consider any $r_1,r_2$ in a bounded set and fixed $L$. Then, the kernel converges as
\begin{equation}
	\lim_{{t}\to\infty}\mathcal{K}^{\rm resc}(r_1,s_1;r_2,s_2)=K^\delta(r_1,s_1;r_2,s_2)
\end{equation}
uniformly for $(s_1,s_2)\in[-L,L]^2$.
\end{prop}

\begin{cor}\label{corBound}
Consider $r_1,r_2$ fixed. For any $L$ there exists ${t}_0$ such that for ${t}>{t}_0$ the bound
\begin{equation}
	\left|\mathcal{K}^{\rm resc}(r_1,s_1;r_2,s_2)\right|\leq \const_L
\end{equation}
holds for all $(s_1,s_2)\in[-L,L]^2$.
\end{cor}

\begin{prop}\label{propK0Bound}
For fixed $r_1,r_2,L$ and $\delta>0$ there exists ${t}_0>0$ such that the estimate
\begin{equation}
	\left|\mathcal{K}_0^{\rm resc}(r_1,s_1;r_2,s_2)\right|\leq \const\cdot e^{-\min\{\delta,1\} s_2}
\end{equation}
holds for any ${t}>{t}_0$ and $s_1,s_2>0$.
\end{prop}

\begin{prop}[Proposition 5.4 of~\cite{FSW13}]\label{propPhiBound}
For fixed $r_1<r_2$ there exists ${t}_0>0$ and $C>0$ such that
\begin{equation}
	\left|\phi_{r_1,r_2}^{\rm resc}(s_1,s_2)\right|\leq Ce^{-|s_1-s_2|}
\end{equation}
\end{prop}

Now we can prove the asymptotic theorem:
\begin{proof}[Proof of Theorem~\ref{thmAsymp}]The joint distributions of the rescaled process $X^{(\delta)}_{t}(r)$ under the measure $\Pb_+$ are given by the Fredholm determinant with series expansion
\begin{equation}\label{eqFredExp}\begin{aligned}
	\Pb_+&\bigg(\bigcap_{k=1}^m\big\{X_{t}^{(\delta)}(r_k)\leq s_k\big\}\bigg)\\
	&\quad=\sum_{N\geq0}\frac{(-1)^N}{N!}\sum_{i_1,\dots,i_N=1}^m\int\prod_{k=1}^N\D \zeta_k\,\Id_{(\zeta_k>\xi_{i_k})}\det_{1\leq k,l\leq N}\left[\mathcal{K}(n_{i_k},\zeta_k;n_{i_l},\zeta_l)\right],
\end{aligned}\end{equation}
where $n_i$ and $\xi_i$ are given in (\ref{eqScaling}). Using the change of variables \mbox{$\sigma_k={t}^{-1/3}(\zeta_k-2{t}-2{t}^{2/3}r_{i_k})$} and a conjugation we obtain
\begin{equation}\label{eqFredExp2}\begin{aligned}
	\eqref{eqFredExp}&=\sum_{N\geq0}\frac{(-1)^N}{N!}\sum_{i_1,\dots,i_N=1}^m\int\prod_{k=1}^N\D \sigma_k\,\Id_{(\sigma_k>s_{i_k})}\\
	&\quad\times\det_{1\leq k,l\leq N}\left[\mathcal{K}^{\rm resc}(r_k,\sigma_k;r_l,\sigma_l)\frac{(1+\sigma_l^2)^{m+1-i_l}}{(1+\sigma_k^2)^{m+1-i_k}}\right],
\end{aligned}\end{equation}
where the fraction inside the determinant is the new conjugation, which does not change the value of the determinant.
Using Corollary~\ref{corBound} and Propositions~\ref{propK0Bound},~\ref{propPhiBound}, we can bound the $(k,l)$-coefficient inside the determinant by
\begin{equation}\label{eqCoeff}
	\const_1\left(e^{-|\sigma_k-\sigma_l|}\Id_{(i_k<i_l)}+e^{-\min\{\delta,1\}\sigma_l}\right)\frac{(1+\sigma_l^2)^{m+1-i_l}}{(1+\sigma_k^2)^{m+1-i_k}},
\end{equation}
assuming the $r_k$ are ordered. The bounds
\begin{equation}\begin{aligned}
	\frac{(1+x^2)^i}{(1+y^2)^j}e^{-|x-y|}&\leq \const_2\frac{1}{1+y^2},& \text{for } i&<j,\\
	\frac{(1+x^2)^i}{(1+y^2)^j}e^{-\min\{\delta,1\} x}&\leq \const_3\frac{1}{1+y^2},& \text{for } j&\geq1,
\end{aligned}\end{equation}
lead to
\begin{equation}
	\eqref{eqCoeff}\leq\const_4\frac{1}{1+\sigma_k^2}.
\end{equation}
Using the Hadamard bound on the determinant, the integrand of \eqref{eqFredExp2} is therefore bounded by
\begin{equation}
	\const_4^NN^{N/2}\prod_{k=1}^N\Id_{(\sigma_k>s_{i_k})}\frac{\D\sigma_k}{1+\sigma_k^2},
\end{equation}
which is integrable. Furthermore,
\begin{equation}
	|\eqref{eqFredExp}|\leq\sum_{N\geq0}\frac{\const_5^NN^{N/2}}{N!},
\end{equation}
which is summable, since the factorial grows like $(N/e)^N$, i.e., much faster than the nominator. Dominated convergence thus allows to interchange the limit ${t}\to\infty$ with the integral and the infinite sum. The pointwise convergence comes from Proposition~\ref{propPointw}, thus
\begin{equation}
	\lim_{{t}\to\infty} \Pb_+\bigg(\bigcap_{k=1}^m\big\{X_{t}^{(\delta)}(r_k)\leq s_k\big\}\bigg)=\det\left(\Id-\chi_s K^\delta\chi_s\right)_{L^2(\{r_1,\dots,r_m\}\times\R)}.
\end{equation}

It remains to show that the convergence carries over to the measure $\Pb$. The identity
\begin{equation}
	\frac{\D s_i}{\D\xi_i}={t}^{-1/3}=\delta^{-1}(1-\rho)
\end{equation}
leads to
\begin{equation}\label{eqShift}
\Pb\bigg(\bigcap_{k=1}^m\big\{X_{t}^{(\delta)}(r_k)\leq s_k\big\}\bigg)=\bigg(1+\frac{1}{\delta}\sum_{i=1}^m\frac{\D}{\D s_i}\bigg)\Pb_+\bigg(\bigcap_{k=1}^m\big\{X_{t}^{(\delta)}(r_k)\leq s_k\big\}\bigg).
\end{equation}
Notice that in \eqref{eqFredExp2}, $s_i$ appears only in the indicator function, so differentiation just results in one of the $\sigma_k$ not being integrated but instead being set to $s_i$. Using the same bounds as before we can again show interchangeability of the limit ${t}\to\infty$ with the remaining integrals and the infinite sum.
\end{proof}

Before showing Propositions~\ref{propPointw} and~\ref{propK0Bound}, we introduce some auxiliary functions and establish asymptotic results for them.
\begin{defin}
Define
\begin{equation}\begin{aligned}
	\alpha_{t}(r,s)&:=\frac{{t}^{1/3}}{2\pi\I}\int_{\I\R}dw\,e^{{t}(w^2-1)/2+\xi(w+1)}(-w)^{n}\\
	&=\frac{{t}^{1/3}}{2\pi\I}\int_{\I\R}dw\,e^{{t}(w^2-1)/2+(2{t}+2{t}^{2/3}r+{t}^{1/3}s)(w+1)}(-w)^{{t}+2{t}^{2/3}r}\\
	\beta_{t}(r,s)&:=\frac{{t}^{1/3}}{2\pi\I}\oint_{\Gamma_0}dz\,e^{-{t}(z^2-1)/2-\xi(z+1)}(-z)^{-n}\\
	&=\frac{{t}^{1/3}}{2\pi\I}\oint_{\Gamma_0}dz\,e^{-{t}(z^2-1)/2-(2{t}+2{t}^{2/3}r+{t}^{1/3}s)(z+1)}(-z)^{-{t}-2{t}^{2/3}r}.	
\end{aligned}\end{equation}
\end{defin}

\begin{lem}\label{lemAlphaLimit}
For fixed $r$ and $L$ the limits
\begin{equation}\begin{aligned}
	\alpha(r,s)&:=\lim_{{t}\to\infty}\alpha_{t}(r,s)=\Ai(r^2+s)e^{-\frac{2}{3}r^3-rs}\\
	\beta(r,s)&:=\lim_{{t}\to\infty}\beta_{t}(r,s)=-\Ai(r^2+s)e^{\frac{2}{3}r^3+rs}
\end{aligned}\end{equation}
hold uniformly for $s\in[-L,L]$.
\end{lem}

\begin{proof}
Let $\He_n(x)$ be the normalized $n$-th order Hermite polynomial with respect to the weight $e^{-x^2/2}$, i.e., satisfying
\begin{equation}
	\int_\R\He_m(x)\He_n(x)e^{-x^2/2}\D x =\delta_{mn}.
\end{equation}
There exist two explicit integral representations for these polynomials:
\begin{equation}\begin{aligned}
	\He_n(x)&=\frac{1}{\I(2\pi)^{3/4}\sqrt{n!}}e^{x^2/2}\int_{\I\R+\e}\D w\,e^{w^2/2-xw}w^{n}\\
	\He_n(x)&=\frac{\sqrt{n!}}{\I(2\pi)^{5/4}}\oint_{\Gamma_0}\D z\,e^{-(z^2/2-xz)}z^{-(n+1)}.
\end{aligned}\end{equation}
From \cite{Sze67}, p. 201, one obtains the asymptotic behaviour
\begin{equation}
	n^{1/12}e^{-x^2/4}\He_n\left(x\right)=\Ai(u)+\Or(n^{-2/3}),\quad x=2\sqrt{n}+n^{-1/6}u
\end{equation}
uniformly for any bounded $u$.

Introducing the shorthands $x=2{t}^{1/2}+2{t}^{1/6}r+{t}^{-1/6}s$, $n={t}+2{t}^{2/3}r$ and applying the change of variables $w\mapsto-w{t}^{-1/2}$, we can write
\begin{equation}\label{eqAlphaRep}
	\alpha_{t}(r,s)=\frac{{t}^{1/3}}{(2\pi)^{1/4}}e^{-{t}/2+{t}^{1/2}x}{t}^{-(n+1)/2}\sqrt{n!}e^{-x^2/2}\He_n(x).
\end{equation}
Using Stirling's approximation and Taylor expansion in the exponents one can further analyze this as
\begin{equation}
	\alpha_{t}(r,s)=e^{-\frac{2}{3}r^3-rs+\Or({t}^{-1/3})}\left(1+\Or({t}^{-1/3})\right)n^{1/12}e^{-x^2/4}\He_n(x),
\end{equation}
with with the error terms being uniform for $s\in[-L,L]$. The observation
\begin{equation}
	u=n^{1/6}(x-2\sqrt{n})\to r^2+s,
\end{equation}
as ${t}\to\infty$, settles the convergence of $\alpha_{t}$.

Using the second integral representation of the Hermite polynomials one can rewrite $\beta_{t}$, too:
\begin{equation}\label{eqBetaRep}
	\beta_{t}(r,s)=-{t}^{1/3}(2\pi)^{1/4}e^{{t}/2-{t}^{1/2}x}{t}^{n/2}(n!)^{-1/2}\He_n(x).
\end{equation}
Analyzing the prefactor as done before finishes the proof.
\end{proof}

\begin{lem}\label{lemAlphaBound}
For fixed $r$, there exist ${t}_0,L$ such that for all ${t}>{t}_0$ and $s>L$ the following bounds hold
\begin{equation}\begin{aligned}
	|\alpha_{t}(r,s)|&\leq e^{-s}\\
	|\beta_{t}(r,s)|&\leq e^{-s}
\end{aligned}\end{equation}
\end{lem}

\begin{proof}
We start by analyzing $\beta_{t}$. Defining functions as
\begin{equation}\begin{aligned}
	f_3(z)&=-(z^2-1)/2-2(z+1)-\ln(-z)\\
	f_2(z)&=-2r(z+1+\ln(-z))\\
	f_1(z)&=-s(z+1),
\end{aligned}\end{equation}
we can write $G(z)={t} f_3(z)+{t}^{2/3}f_2(z)+{t}^{1/3}f_1(z)$, leading to
\begin{equation}
	\beta_{t}(r,s)=\frac{{t}^{1/3}}{2\pi\I}\oint_{\Gamma_0}\D z\, e^{G(z)}.
\end{equation}
\begin{figure}
 \centering
 \psfrag{-1}[lb]{$-1$}
 \psfrag{omega}[lb]{$\omega$}
 \psfrag{Gamma}[lb]{$\Gamma$}
 \psfrag{theta}[lb]{$\theta$}
 \psfrag{R}[lb]{$R$}
 \includegraphics[height=5cm]{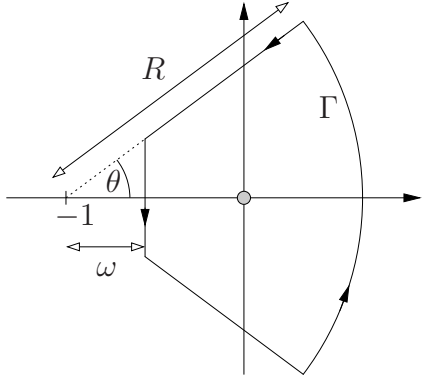}
 \caption{The contour $\Gamma=\gamma_1\cup\gamma_2(R)\cup\overline{\gamma_2(R)}\cup\gamma_3(R)$ used for obtaining the uniform bounds.}
 \label{figboundContour}
\end{figure}

Define a new parameter $\omega$ given by
\begin{equation}\label{eqOmega}
	\omega=\min\left\{{t}^{-1/3}\sqrt{s},\e\right\},
\end{equation}
for some small, positive $\e$ chosen in the following, and let $\theta\in(\pi/6,\pi/4)$. As shown in Figure~\ref{figboundContour}, we change the contour $\Gamma_0$ to \mbox{$\gamma_1\cup\gamma_2(R)\cup\overline{\gamma_2(R)}\cup\gamma_3(R)$}, with
\begin{equation}\begin{aligned}
	\gamma_1&=\{-1+\omega(1+\I u\tan\theta), u\in[-1,1] \}\\
	\gamma_2(R)&=\{-1+ue^{\I\theta}, u\in[\omega/\cos\theta,R]\}\\
	\gamma_3(R)&=\{-1+Re^{\I u}, u\in[-\theta,\theta]\}.
\end{aligned}\end{equation}
Since we will only estimate the absolute value of the integrals, the direction of integration does not matter.

If ${t}$ and $s$ are fixed, the integrand is dominated by the $\exp(-z^2)$ term for large $|z|$. Thus the contribution coming from $\gamma_3(R)$ converges to $0$ as $R\to\infty$. With $\gamma_2=\lim_{R\to\infty}\gamma_2(R)$ our choice for the contour of integration is now \mbox{$\gamma_1\cup\gamma_2\cup\overline{\gamma_2}$}.

We start by analyzing
\begin{equation}\label{eqAlphFac}
	\frac{{t}^{1/3}}{2\pi\I}\int_{\gamma_1}\D z\, e^{G(z)}=e^{G(z_0)}\frac{{t}^{1/3}}{2\pi}\int_{[-\omega\tan\theta,\omega\tan\theta]}\D u\, e^{G(z_0+\I u)-G(z_0)},
\end{equation}
where $z_0=-1+\omega$.

Let us consider the prefactor $e^{G(z_0)}$ at first. Since $\omega$ is small we can use Taylor expansion, as well as \eqref{eqOmega}, to obtain the bounds
\begin{equation}\begin{aligned}
	{t} f_3(z_0)&={t}\left(\omega^3/3+\Or(\omega^4)\right)\leq\frac{1}{3}\omega s{t}^{1/3}\left(1+\Or(\e)\right)\\
	{t}^{2/3}f_2(z_0)&={t}^{2/3}r\left(\omega^2+\Or(\omega^3)\right)\leq\omega \sqrt{s}{t}^{1/3}|r|\left(1+\Or(\e)\right)\\
	{t}^{1/3}f_1(z_0)&=-\omega s{t}^{1/3}.
\end{aligned}\end{equation}
All error terms are to be understood uniformly in $s,{t},r$. The $f_1$ term dominates both $f_2$, if $L$ is chosen large enough, and $f_3$, for $\e$ being small. This results in
\begin{equation}\label{eqG0bound}
	e^{G(z_0)}\leq e^{-\frac{1}{2}\omega s{t}^{1/3}}\leq e^{-\frac{1}{2}s^{3/2}}.
\end{equation}

To show convergence of the integral part of \eqref{eqAlphFac} we first bound the real part of the exponent:
\begin{equation}\begin{aligned}
&\Re\left(G(z_0+\I u)-G(z_0)\right)\\
=&\Re\bigg[{t}\left(\frac{u^2-2z_0\I u}{2}-2\I u-\ln\frac{z_0+\I u}{z_0}\right)+{t}^{2/3}\cdot 2r\left(\I u+\ln\frac{z_0+\I u}{z_0}\right)-{t}^{1/3}s\I u\bigg]\\
=&{t}\left(\frac{u^2}{2}-\frac{1}{2}\ln\left(1+\frac{u^2}{z_0^2}\right)\right)+{t}^{2/3}r\ln\left(1+\frac{u^2}{z_0^2}\right)\\
\leq&{t}\frac{u^2}{2}\left(1-\frac{1}{z_0^2}+\frac{u^2}{2z_0^4}\right)+{t}^{2/3}r\frac{u^2}{z_0^2}=:-\eta{t}^{2/3}u^2.
\end{aligned}\end{equation}
$\eta$ satisfies:
\begin{equation}\begin{aligned}
\eta&=\frac{{t}^{1/3}}{2}\left(\frac{1}{(1-\omega)^2}-1-\frac{u^2}{2(1-\omega)^4}\right)-\frac{r}{(1-\omega)^2}\\
&={t}^{1/3}\omega\left(1+\Or(\omega)\right)-r\left(1+\Or(\omega)\right),
\end{aligned}\end{equation}
where we used $|u|<\omega$. Given any $\e$ we can now choose both $L$ and ${t}_0$ large, such that the first term dominates. Consequently $\eta$ will be bounded from below by some positive constant $\eta_0$. The integral contribution coming from $\gamma_1$ can thus be bounded as
\begin{equation}\begin{aligned}
	|\eqref{eqAlphFac}|&= e^{G(z_0)}\frac{{t}^{1/3}}{2\pi}\left|\int_{[-\omega\tan\theta,\omega\tan\theta]}\D u\,e^{G(z_0+iu)-G(z_0)}\right|\\
	&\leq e^{-\frac{1}{2}s^{3/2}}\frac{{t}^{1/3}}{2\pi}\int_\R \D u\,e^{-\eta_0{t}^{2/3}u^2}=e^{-\frac{1}{2}s^{3/2}}\frac{1}{2\pi}\int_\R \D u\,e^{-\eta_0u^2}\\
	&=e^{-\frac{1}{2}s^{3/2}}\frac{1}{2\sqrt{\pi\eta_0}}\leq e^{-s}.
\end{aligned}\end{equation}

Finally we need a corresponding bound on the $\gamma_2$ contribution to the integral. By symmetry this case covers also the contour $\overline{\gamma_2}$. Write
\begin{equation}\label{eqAlphFac2}
	\frac{{t}^{1/3}}{2\pi\I}\int_{\gamma_2}\D z\, e^{G(z)}=e^{G(z_1)}\frac{{t}^{1/3}e^{\I\theta}}{2\pi\I}\int_{\R_+}\D u\, e^{G(z_1+ue^{\I\theta})-G(z_1)},
\end{equation}
with $z_1=-1+\omega(1+\I\tan\theta)$. From the previous estimates one easily gets
\begin{equation}
	\left|e^{G(z_1)}\right|\leq e^{G(z_0)}\leq e^{-\frac{1}{2}s^{3/2}},
\end{equation}
so the remaining task is to show boundedness of the integral part of \eqref{eqAlphFac2}.

At first notice that the real part of the $f_1$ contribution in the exponent is negative, so we can omit it, avoiding the problem of large $s$. By elementary calculus, we have for all $u\geq\omega/\cos\theta$,
\begin{equation}
	\frac{\D}{\D u}\Re\left(f_3(-1+ue^{\I\theta})\right)<0,
\end{equation}
that is, $\gamma_2$ is a steep descent curve for $f_3$. We can therefore restrict the contour to a neighbourhood of the critical point $z_1$, which we choose of magnitude $\delta$, at the expense of an error of order $\Or(e^{-\const_\delta {t}})$:
\begin{equation}\label{eq3.43}\begin{aligned}
	\left|\frac{{t}^{1/3}}{2\pi}\int_{\R_+}\D u\, e^{G(z_1+ue^{\I\theta})-G(z_1)}\right|&\leq \left(1+\Or(e^{-\const_\e {t}})\right)\int_0^\delta\D u\, \left|e^{{t}\hat{f}_3(ue^{\I\theta})+{t}^{2/3}\hat{f}_2(ue^{\I\theta})}\right|
\end{aligned}\end{equation}
where $\hat{f}_i(z)=f_i(z_1+z)-f_i(z_1)$. Taylor expanding these functions leads to
\begin{equation}\begin{aligned}
	\Re({t}\hat{f}_3(ue^{\I\theta}))&={t} \Re(e^{3\I\theta})u\frac{\omega^2}{\cos^2\theta}\left(1+\Or(\delta)\right)\left(1+\Or(\e)\right)\\
&\leq-\chi_3{t}^{1/3}\omega \cdot{t}^{2/3}u\omega\\	\Re({t}^{2/3}\hat{f}_2(ue^{\I\theta}))&=2\Re(e^{2\I\theta}){t}^{2/3}ru\frac{\omega}{\cos\theta}\left(1+\Or(\delta)\right)\left(1+\Or(\e)\right)\\
&\leq\chi_2|r|\cdot{t}^{2/3}u\omega,
\end{aligned}\end{equation}
for some positive constants $\chi_2$, $\chi_3$, by choosing $\delta$ and $\e$ small enough. For large $L$ and ${t}_0$, $-\chi_3{t}^{1/3}\omega$ dominates over $\chi_2|r|$, so we can further estimate:
\begin{equation}
	\int_0^\delta\D u\, \left|e^{{t}\hat{f}_3(ue^{\I\theta})+{t}^{2/3}\hat{f}_2(ue^{\I\theta})+\hat{f}_0(ue^{\I\theta})}\right|\leq\int_0^\infty\D u\, e^{-\chi_3{t}\omega^2 u/2}\leq\frac{2}{\chi_3{t}\omega^2}.
\end{equation}
This settles the uniform exponential bound on $\beta_{t}$.

Combining \eqref{eqAlphaRep} and \eqref{eqBetaRep} one obtains
\begin{equation}
	\frac{\alpha_{t}(r,s)}{\beta_{t}(r,s)}=f_{t}(r)e^{-2rs-s^2{t}^{-1/3}}
\end{equation}
for some function $f_{t}(r)$. From the convergence of $\alpha_{t}$ and $\beta_{t}$ it is clear that $f_{t}$ converges, too. Since we already know that $\beta_{t}$ is uniformly bounded by a constant times $e^{-s^{3/2}/2}$, the exponential bound on $\alpha_{t}$ follows.
\end{proof}

\begin{proof}[Proof of Proposition~\ref{propPointw}]
Regarding the first part of the kernel, we notice, that $n_i=0$ does not appear in our scaling, so we can use the formula (for $n_2>n_1$)
\begin{equation}
\phi_{n_1,n_2}(\xi_1,\xi_2)=(\phi_{n_1}*\dots\phi_{n_2-1})(\xi_1,\xi_2)=\frac{(\xi_2-\xi_1)^{n_2-n_1-1}}{(n_2-n_1-1)!}\Id_{\xi_1\leq \xi_2}.
\end{equation}
This is the same function as in \cite{FSW13}, proof of Proposition 5.1, so the limit
\begin{equation}\begin{aligned}
	\lim_{t\to\infty}\phi^{\rm resc}_{r_1,r_2}(s_1,s_2)=
		\frac{1}{\sqrt{4\pi (r_2-r_1)}}e^{-(s_2-s_1)^2/4(r_2-r_1)}\Id(r_1<r_2),
\end{aligned}\end{equation} is not proven here.

The different parts of the remaining kernel can be rewritten as integrals over the previously defined functions $\alpha$ and $\beta$. For $\widetilde{\mathcal{K}}$, choose the contours in such a way that $\Re(z-w)>0$ is ensured.
\begin{equation}\begin{aligned}
&{t}^{1/3}e^{\xi_1-\xi_2}\widetilde{\mathcal{K}}(n_1,\xi_1;n_2,\xi_2)\\
=&\frac{{t}^{1/3}}{(2\pi\I)^2}\int_{\I\R-\e}\D w\oint_{\Gamma_0}\D z\frac{e^{{t} w^2/2+\xi_1(w+1)}}{e^{{t} z^2/2+\xi_2 (z+1)}}\frac{(-w)^{n_1}}{(-z)^{n_2}}\frac{1}{w-z}\\
	=&\frac{-{t}^{1/3}}{(2\pi\I)^2}\int_{\I\R-\e}\D w\oint_{\Gamma_0}\D z\frac{e^{{t} (w^2-1)/2+\xi_1(w+1)}}{e^{{t} (z^2-1)/2+\xi_2 (z+1)}}\frac{(-w)^{n_1}}{(-z)^{n_2}}\int_0^\infty \D x\,{t}^{1/3}e^{-{t}^{1/3}x(z-w)}\\
	=&-\int_0^\infty \D x\,\alpha_{t}(r_1,s_1+x)\beta_{t}(r_2,s_2+x)
\end{aligned}\end{equation}
Also rewrite $\mathpzc{f}$ as follows:
\begin{equation}\begin{aligned}
&e^{-{t}/2+\xi_1}\mathpzc{f}(n_1,\xi_1)=\frac{1}{2\pi\I}\int_{\I\R-\e}\D w\,\frac{e^{{t} (w^2-1)/2+\xi_1(w+1)}(-w)^{n_1}}{w+1}\\
	&=1+\frac{1}{2\pi\I}\int_{\I\R-\e-1}\D w\,\frac{e^{{t} (w^2-1)/2+\xi_1(w+1)}(-w)^{n_1}}{w+1}\\
	&=1-\frac{1}{2\pi\I}\int_{\I\R-\e-1}\D w\,e^{{t} (w^2-1)/2+\xi_1(w+1)}(-w)^{n_1}\int_0^\infty \D x\, {t}^{1/3}e^{{t}^{1/3}x(w+1)}\\&=1-\int_0^\infty \D x\, \alpha_{{t}}(r_1,s_1+x).
\end{aligned}\end{equation}
Similarly,
\begin{equation}\begin{aligned}
	e^{{t}/2-\xi_2}\mathpzc{g}(n_2,\xi_2)&={\rm Res}_{\mathpzc{g},-\rho}+\int_0^\infty \D x\, \beta_{{t}}(r_2,s_2+x)e^{\delta x},
\end{aligned}\end{equation}
with
\begin{equation}
	{\rm Res}_{\mathpzc{g},-\rho}=e^{{t}^{2/3}\delta-{t}^{1/3}\delta^2/2-\xi_2{t}^{-1/3}\delta}(1-{t}^{-1/3}\delta)^{-n_2}.
\end{equation}
The residuum satisfies the limit
\begin{equation}\label{eqResLimit}
	\lim_{{t}\to\infty}{\rm Res}_{\mathpzc{g},-\rho}=e^{\delta^3/3+r_2\delta^2-s_2\delta}
\end{equation}
uniformly in $s_2$.
The prefactor of the last part of the kernel is simply \mbox{${t}^{1/3}(1-\rho)=\delta$}.
Combining all these equations gives
\begin{equation}\begin{aligned}\label{eq26}
&\mathcal{K}_0^{\rm resc}(r_1,s_1;r_2,s_2)=-\int_0^\infty \D x\,\alpha_{t}(r_1,s_1+x)\beta_{t}(r_2,s_2+x)\\
	&+\delta\left(1-\int_0^\infty \D x\, \alpha_{{t}}(r_1,s_1+x)\right)\left({\rm Res}_{\mathpzc{g},-\rho}+\int_0^\infty \D x\, \beta_{{t}}(r_2,s_2+x)e^{\delta x}\right).
\end{aligned}\end{equation}
Using the previous lemmas we can deduce compact convergence of the kernel. Indeed (omitting the $r$-dependence for greater clarity) we can write:
\begin{equation}\begin{aligned}\label{eqSupK0}
	\sup_{s_1,s_2\in[-L,L]}&\left|\int_0^\infty \D x\,\alpha_{t}(s_1+x)\beta_{t}(s_2+x)-\int_0^\infty \D x\,\alpha(s_1+x)\beta(s_2+x)\right|\\
	\leq&\int_0^\infty \D x\,\sup_{s_1,s_2\in[-L,L]}\left|\alpha_{t}(s_1+x)\beta_{t}(s_2+x)-\alpha(s_1+x)\beta(s_2+x)\right|.
\end{aligned}\end{equation}
By Lemma~\ref{lemAlphaLimit} the integrand converges to zero for every $x>0$. Using Lemma~\ref{lemAlphaBound} we can bound it by $\const\cdot e^{-2x}$, thus ensuring that \eqref{eqSupK0} goes to zero, i.e., $\widetilde{\mathcal{K}}$ converges compactly. In the same way we can show the convergence of $\mathpzc{f}$ and $\mathpzc{g}$. Applying the limit in \eqref{eq26} and inserting the expressions for $\alpha$ and $\beta$ finishes the proof.
\end{proof}

\begin{proof}[Proof of Proposition~\ref{propK0Bound}]
Since the convergence \eqref{eqResLimit} is uniform in $s_2$ we can deduce
\begin{equation}
	\left|{\rm Res}_{\mathpzc{g},-\rho}\right|\leq \const_1\cdot e^{-s_2\delta}.
\end{equation}
Inserting this as well as the bounds from Lemma~\ref{lemAlphaBound} into \eqref{eq26} results in
\begin{equation}\begin{aligned}
		\left|\mathcal{K}_0^{\rm resc}(r_1,s_1;r_2,s_2)\right|&\leq\int_0^\infty dx\,e^{-(s_1+x)}e^{-(s_2+x)}+\delta\left(1+\int_0^\infty dx\, e^{-(s_1+x)}\right)\\&\quad\times\left(\const_1\cdot e^{-s_2\delta}+\int_0^\infty dx\, e^{-(s_2+x)}e^{\delta x}\right)\\
		&=\frac{1}{2}e^{-(s_1+s_2)}+\delta\left(1+e^{-s_1}\right)\left(\const_1\cdot e^{-s_2\delta}+\frac{e^{-s_2}}{1-\delta}\right)\\
		&\leq\const\cdot e^{-\min\{\delta,1\} s_2}.
\end{aligned}\end{equation}
\end{proof}

\newpage
\section{Path-integral style formula}\label{SectPathIntegrals}
Using the results from \cite{BCR13} we can transform the formula for the multidimensional probability distribution of the finite-step Airy$_{\rm stat}$ process from the current form involving a Fredholm determinant over the space \mbox{$L^2(\{r_1,\dots,r_m\}\times\R)$} into a path-integral style form, where the Fredholm determinant is over the simpler space $L^2(\R)$. The result of \cite{BCR13} can not be applied at the stage of finite time as one of the assumption is not satisfied.
\begin{prop}\label{propPathInt}
For any parameters $\chi_k\in\R$, $1\leq k\leq m$, satisfying
\begin{equation}
  0<\chi_m<\dots<\chi_2<\chi_1<\max_{i<j}\left\{r_j-r_i,\delta\right\},
\end{equation}
define the multiplication operator $(M_{r_i}f)(x)=m_{r_i}(x)f(x)$, with
\begin{equation}
 m_{r_i}(x)=\begin{cases}e^{-\chi_ix} &\mbox{for } x\geq0 \\ e^{x^2} &\mbox{for } x<0.\end{cases}
\end{equation}
Writing $K^\delta_{r_i}(x,y):=K^\delta(r_i,x;r_i,y)$, the finite-dimensional distributions of the finite-step Airy$_{\rm stat}$ process are given by
\begin{equation}\label{PathIntForm}\begin{aligned}
	\Pb &\bigg(\bigcap_{k=1}^m\big\{\mathcal{A}_{\rm stat}^{(\delta)}(r_k)\leq s_k\big\}\bigg)=\bigg(1+\frac{1}{\delta}\sum_{i=1}^m\frac{\D}{\D s_i}\bigg)
\det\Big(\Id+M_{r_1} Q M_{r_1}^{-1}\Big)_{L^2(\R)},
\end{aligned}\end{equation}
with
\begin{equation}
Q=-K^{\delta}_{r_1}+\bar{P}_{s_1}V_{r_1,r_2}\bar{P}_{s_2}\cdots V_{r_{m-1},r_m}\bar{P}_{s_m}V_{r_m,r_1}K^\delta_{r_1},
\end{equation}
where $\bar{P}_s=\Id-P_s$ denotes the projection operator on $(-\infty,s)$.
\end{prop}
\begin{remark}\label{remV}
 The operator $V_{r_j,r_i}$ for $r_i<r_j$ is defined only on the range of $K^\delta_{r_i}$ and acts on it in the following way:
 \begin{equation}
   V_{r_j,r_i}K_{r_i,r_k}=K_{r_j,r_k},\qquad
   V_{r_j,r_i}f_{r_i}=f_{r_j}.
 \end{equation}
In particular, we have also $V_{r_j,r_i}\mathbf{1}=\mathbf{1}$.
\end{remark}

\begin{proof}
We will denote conjugations by the operator $M$ by a hat in the following way:
\begin{equation}
 \begin{aligned}
  \widehat{V}_{r_i,r_j}&=M_{r_i}V_{r_i,r_j}M^{-1}_{r_j}, &\quad \widehat{f}_{r_i}&=M_{r_i}f_{r_i},\\
  \widehat{K}^\delta_{r_i}&=M_{r_i}K^\delta_{r_i}M^{-1}_{r_i}, &\quad \widehat{g}_{r_i}&=g_{r_i}M_{r_i}^{-1},\\
  \widehat{K}_{r_i,r_j}&=M_{r_i}K_{r_i,r_j}M^{-1}_{r_j}.
 \end{aligned}
\end{equation}
Applying the conjugation also in the determinant in \eqref{eq3.6}, the identity we have to show is:
\begin{equation}
\begin{aligned}
  \det&\left(\Id-\chi_s \widehat{K}^\delta\chi_s\right)_{L^2(\{r_1,\dots,r_m\}\times\R)}\\&= \det\Big(\Id-\widehat{K}^{\delta}_{r_1}+\bar{P}_{s_1}\widehat{V}_{r_1,r_2}\bar{P}_{s_2}\cdots \widehat{V}_{r_{m-1},r_m}\bar{P}_{s_m}\widehat{V}_{r_m,r_1}\widehat{K}^\delta_{r_1}\Big)_{L^2(\R)}
\end{aligned}
\end{equation}
This is done by applying Theorem 1.1 \cite{BCR13}.

It has three groups of assumptions we have to prove. We merged them into two by choosing the multiplication operators of Assumption 3 to be the identity.
\paragraph{Assumption 1}
\renewcommand{\theenumi}{\roman{enumi}}
\renewcommand{\labelenumi}{(\theenumi)}
\begin{enumerate}
	\item The operators $P_{s_i}\widehat{V}_{r_i,r_j}$, $P_{s_i}\widehat{K}^\delta_{r_i}$, $P_{s_i}\widehat{V}_{r_i,r_j}\widehat{K}^\delta_{r_j}$ and $P_{s_j}\widehat{V}_{r_j,r_i}\widehat{K}^\delta_{r_i}$ for $r_i<r_j$ preserve $L^2(\R)$ and are trace class in $L^2(\R)$.
	\item The operator $\widehat{V}_{r_i,r_1}\widehat{K}^{\delta}_{r_1}-\bar{P}_{s_i}\widehat{V}_{r_i,r_{i+1}}\bar{P}_{s_{i+1}}\cdots \widehat{V}_{r_{m-1},r_m}\bar{P}_{s_m}\widehat{V}_{r_m,r_1}\widehat{K}^\delta_{r_1}$ preserves $L^2(\R)$ and is trace class in $L^2(\R)$.
\end{enumerate}
\paragraph{Assumption 2}
\begin{enumerate}
	\item Right-invertibility: $\widehat{V}_{r_i,r_j}\widehat{V}_{r_j,r_i}\widehat{K}^\delta_{r_i}=\widehat{K}^\delta_{r_i}$
	\item Semigroup property: $\widehat{V}_{r_i,r_j}\widehat{V}_{r_j,r_k}=\widehat{V}_{r_i,r_k}$
	\item Reversibility relation: $\widehat{V}_{r_i,r_j}\widehat{K}^\delta_{r_j}=\widehat{K}^\delta_{r_i}\widehat{V}_{r_i,r_j}$
\end{enumerate}

The semigroup property is clear. To see the reversibility relation, start from the contour integral representation \eqref{contInt} of $K_{r_j,r_j}$ and $f_{r_j}$ and use the Gaussian identity:
\begin{equation}\begin{aligned}
	\int_\R dz \frac{1}{\sqrt{4\pi (r_j-r_i)}}e^{-(z-x)^2/4(r_j-r_i)}e^{-r_jW^2+zW}=e^{-r_iW^2+xW}.
\end{aligned}\end{equation}
This results in $\widehat{V}_{r_i,r_j}\widehat{K}^\delta_{r_j}=\widehat{K}_{r_i,r_j}+\delta \widehat{f}_{r_i}\otimes \widehat{g}_{r_j}$.
On the other hand we have
\begin{equation}\begin{aligned}
	\int_\R dz \frac{1}{\sqrt{4\pi (r_j-r_i)}}e^{-(z-y)^2/4(r_j-r_i)}e^{r_iZ^2-zZ}=e^{r_jZ^2-yZ},
\end{aligned}\end{equation}
so $\widehat{K}^\delta_{r_i}\widehat{V}_{r_i,r_j}=\widehat{K}_{r_i,r_j}+\delta \widehat{f}_{r_i}\otimes \widehat{g}_{r_j}$, which proves Assumption 2 (iii). Noticing Remark~\ref{remV}, the right-invertibility follows immediately.

Assumption 1 (ii) can be deduced from Assumption 1 (i) as shown in Remark~3.2, \cite{BCR13}. Using the previous identities we thus are left to show that the three operators $P_{s_i}\widehat{V}_{r_i,r_j}$, for $r_i<r_j$, as well as $P_{s_i}\widehat{K}_{r_i,r_j}$ and $P_{s_i}\widehat{f}_{r_i}\otimes \widehat{g}_{r_j}$, for arbitrary $r_i$, $r_j\in\R$, are all $L^2$-bounded and trace class.

First notice that $V_{r_i,r_j}(x,y)=V_{0,r_j-r_i}(-x,-y)$. Using the shorthand $r=r_j-r_i$ and inserting this into the integral representation \eqref{eqVRep} of $V$ we have
\begin{equation}
 V_{r_i,r_j}(x,y)=e^{\frac{2}{3}r^3}\int_\R\D \lambda\, \Ai(-x+\lambda)e^{r(-y+\lambda)}\Ai(r^2-y+\lambda)=\left(V^{(1)}V^{(2)}_r\right)(x,y),
\end{equation}
with the new operators
\begin{equation}\begin{aligned}
 V^{(1)}(x,y)&=\Ai(-x+y)\\
 V^{(2)}_r(x,y)&=e^{\frac{2}{3}r^3}e^{r(x-y)}\Ai(r^2+x-y).
\end{aligned}\end{equation}
Introducing yet another operator, $(Nf)(x)=\exp\left(-(\chi_i+\chi_j)x/2\right)f(x)$, we can write
\begin{equation}
 P_{s_i}\widehat{V}_{r_i,r_j}=(P_{s_i}M_{r_i}V^{(1)}N^{-1})(NV^{(2)}_rM_{r_j}^{-1}).
\end{equation}
The Hilbert-Schmidt norm of the first factor is given by
\begin{equation}\begin{aligned}
	\int_{\R^2}\D x\,\D y\, &\left|(P_{s_i}M_{r_i}V^{(1)}N^{-1})(x,y)\right|^2\\&=\int_{s_1}^{\infty}\D x\,\int_\R\D y\,m^2_{r_i}(x)\Ai^2(-x+y)e^{(\chi_i+\chi_j)y}\\
	&=\int_{s_1}^{\infty}\D x\,m^2_{r_i}(x)e^{(\chi_i+\chi_j)x}\int_\R\D z\,\Ai^2(z)e^{(\chi_i+\chi_j)z}.
\end{aligned}\end{equation}
The asymptotic behaviour of the Airy function and the inequalities \mbox{$\chi_i>\chi_j>0$} imply that both integrals are finite. Similarly,
\begin{equation}\begin{aligned}
	\int_{\R^2}\D x\,\D y\, &\left|(NV^{(2)}_rM_{r_j}^{-1})(x,y)\right|^2\\
	&=e^{\frac{4}{3}r^3}\int_{\R^2}\D x\,\D y\, e^{-(\chi_i+\chi_j)x}e^{2r(x-y)}\Ai^2(r^2+x-y)m^{-2}_{r_j}(y)\\
	&=e^{\frac{4}{3}r^3}\int_\R\D z\,e^{-(\chi_i+\chi_j)z}e^{2rz}\Ai^2(r^2+z)\int_\R\D y\,m^{-2}_{r_j}(y)e^{-(\chi_i+\chi_j)y}<\infty,
\end{aligned}\end{equation}
where we used $2r>\chi_i+\chi_j$ as well.
As a product of two Hilbert-Schmidt operators, $P_{s_i}\widehat{V}_{r_i,r_j}$ is thus $L^2$-bounded and trace class.

We decompose the operator $\widehat{K}_{r_i,r_j}$ as
\begin{equation}
 P_{s_i}\widehat{K}_{r_i,r_j}=(P_{s_i}M_{r_i}K'_{-r_i}P_0)(P_0K'_{r_j}M_{r_j}^{-1})
\end{equation}
where
\begin{equation}
 K'_r(x,y)=e^{\frac{2}{3}r^3}e^{r(x+y)}\Ai(r^2+x+y).
\end{equation}
Again, we bound the Hilbert-Schmidt norms,
\begin{equation}\begin{aligned}
	\int_{\R^2}\D x\,\D y\, &\left|(P_{s_i}M_{r_i}K'_{-r_i}P_0)(x,y)\right|^2\\
	&=e^{-\frac{4}{3}r_i^3}\int_{s_i}^\infty\D x\,\int_0^\infty\D y\,m^2_{r_i}(x)e^{-2r_i(x+y)}\Ai^2(r_i^2+x+y)\\
	&\leq e^{-\frac{4}{3}r_i^3}\int_{s_i}^\infty\D x\,m^2_{r_i}(x)\int_{s_i}^\infty\D z\,e^{-2r_iz}\Ai^2(r_i^2+z)<\infty,
\end{aligned}\end{equation}
as well as
\begin{equation}\begin{aligned}
	\int_{\R^2}\D x\,\D y\, &\left|(P_0K'_{r_j}M_{r_j}^{-1})(x,y)\right|^2\\
	&=e^{\frac{4}{3}r_j^3}\int_0^\infty\D x\,\int_\R\D y\,e^{2r_j(x+y)}\Ai^2(r_j^2+x+y)m^{-2}_{r_j}(y)\\
	&= e^{\frac{4}{3}r_j^3}\int_\R\D y\,m^{-2}_{r_j}(y)\int_y^\infty\D z\,e^{2r_jz}\Ai^2(r_j^2+z).
\end{aligned}\end{equation}
The superexponential decay of the Airy function implies that for every \mbox{$c_1>|r_j|$} we can find $c_2$ such that $e^{2r_jz}\Ai^2(r_j^2+z)\leq c_2e^{-c_1z}$. This proves finiteness of the integrals.

Regarding the last operator, start by decomposing it as
\begin{equation}
 P_{s_i}\widehat{f}_{r_i}\otimes \widehat{g}_{r_j}=(P_{s_i}\widehat{f}_{r_i}\otimes\phi)(\phi \otimes \widehat{g}_{r_j})
\end{equation}
for some function $\phi$ with $L^2$-norm $1$. Next, notice that
\begin{equation}\begin{aligned}
	\int_{\R^2}\D x\,\D y\, &\left|(P_{s_i}M_{r_i}f_{r_i}\otimes\phi)(x,y)\right|^2=\int_{s_i}^\infty\D x\,m^2_{r_i}(x)f^2_{r_i}(x)
\end{aligned}\end{equation}
It is easy to see that $\lim_{s\to\infty}f_{r_i}(s)=1$, so $f_{r_i}$ is bounded on the area of integration. But then the $m^2_{r_i}$ term ensures the decay, implying that the integral is finite. Furthermore,
\begin{equation}\begin{aligned}
	\int_{\R^2}\D x\,\D y\, &\left|(\phi \otimes g_{r_j}M_{r_j}^{-1})(x,y)\right|^2=\int_\R\D y\,m^{-2}_{r_j}(y)g^2_{r_j}(y).
\end{aligned}\end{equation}
Analyzing the asymptotic behaviour of $g_{r_j}$ we see that for large positive arguments, the first part decays exponentially with rate $-\delta$ and the second part even superexponentially. $\delta>\chi_j$ thus gives convergence on the positive half-line. For negative arguments, it is sufficient to see that $g_{r_j}$ does not grow faster than exponentially.
\end{proof}

\section{Analytic continuation - Proof of Theorem~\ref{thmAsymp0}}\label{SectAnCont}
First of all let us show that the choice of $x_0(0)=0$ is asymptotically irrelevant. Denote by $X_{t}^{(0)}(r)$ the rescaled process as in (\ref{eqScaledProcess}), where $x_0(0)=0$, and $X_{t}(r)$ the rescaled process as in (\ref{eqScaledProcessOriginal}), where $-x_0(0)\sim \exp(1)$. This corresponds to a finite shift of the system, which is therefore irrelevant in the large time limit.
\begin{lem} \label{lemma7.1} For any $m\in\N$, $r_1<r_2<\ldots<r_m$ and $s_1,\ldots,s_m$, it holds
\begin{equation}
\lim_{{t}\to\infty}\Pb\bigg(\bigcap_{k=1}^m\big\{X_{t}(r_k)\leq s_k\big\}\bigg)
=\lim_{{t}\to\infty}\Pb\bigg(\bigcap_{k=1}^m\big\{X_{t}^{(0)}(r_k)\leq s_k\big\}\bigg)
\end{equation}
\end{lem}
\begin{proof}
We can construct the processes $x_n^{(0)}$ and $x_n$ on the same probability space so that, for any $n\in\Z$, $x_n^{(0)}(t)=x_n(t)-x_0(0)$ and with $x_0(0)$ being independent of $x_n(t)-x_0(0)$. After scaling we have $X_{t}^{(0)}(r)=X_{t}(r)-x_0(0)t^{-1/3}$. As $x_0(0)t^{-1/3}$ converges to $0$ is distribution, the result follows.
\end{proof}

We know from Theorem~\ref{thmAsymp} and Proposition~\ref{propPathInt} that:
\begin{equation}
\lim_{{t}\to\infty}\Pb\bigg(\bigcap_{k=1}^m\big\{X_{t}^{(\delta)}(r_k)\leq s_k\big\}\bigg)=\bigg(1+\frac{1}{\delta}\sum_{i=1}^m\frac{\D}{\D s_i}\bigg)\det(\Id-\widehat{\mathcal{P}}\widehat{K}^\delta_{r_1}).
\end{equation}
In this section we prove the main Theorem~\ref{thmAsymp0} by extending this equation to $\delta=0$. The right hand side can actually be analytically continued for all $\delta\in\R$ (see Proposition~\ref{propAnalyt}). Additionally we have to show that the left hand side is continuous at $\delta=0$. This proof relies mainly on Proposition~\ref{propExPoint}, which gives a bound on the exit point of the maximizing path from the lower boundary in the last passage percolation model.

\begin{proof}[Proof of Theorem~\ref{thmAsymp0}]
We adopt the point of view of last passage percolation discussed in Section~\ref{SectLPP}. The superscripts of $x$, $L$ and $w$ indicate the choice of $\rho$, while $\lambda$ is always fixed at $1$. It is clear that for any path $\vec{\pi}$ the weight $w^{(\rho)}(\vec{\pi})$ is non-decreasing in $\rho$. But then the supremum is non-decreasing, too, and:
\begin{equation}\label{eq5.6}
 x_n^{(\rho)}(t)\leq x_n^{(1)}(t),
\end{equation}
for $\rho<1$.
We know that there exists a unique maximizing path \mbox{$\vec{\pi}^*\in\Pi(0,0;t;n)$}. We can therefore define $Z_n(t):=s_0^*$, the exit point from the lower boundary specifically with $\rho=1$. We want to derive the inequality
\begin{equation}\label{eq5.7}
   x_n^{(1)}(t)\leq x_n^{(\rho)}(t)+(1-\rho)Z_n(t).
\end{equation}
This can be seen as follows:
\begin{equation}\begin{aligned}
 L^{(1)}_{(0,0)\to(t,n)}-(1-\rho)Z_n(t)&=\sup_{\vec{\pi}\in\Pi(0,0;t,n)}w^{(1)}(\vec{\pi})-(1-\rho)Z_n(t)\\
 &=w^{(1)}(\vec{\pi}^*)-(1-\rho)s_0^*=w^{(\rho)}(\vec{\pi}^*).
\end{aligned}\end{equation}
Note that $\vec{\pi}^*$ maximizes $w^{(1)}(\vec\pi)$ and not necessarily $w^{(\rho)}(\vec{\pi})$. In particular we have
\begin{equation}
 w^{(\rho)}(\vec{\pi}^*)\leq\sup_{\vec{\pi}\in\Pi(0,0;t,n)}w^{(\rho)}(\vec{\pi})= L^{(\rho)}_{(0,0)\to(t,n)}.
\end{equation}
Combining the last two equations results in \eqref{eq5.7}.

\eqref{eq5.6} and \eqref{eq5.7} imply that for the rescaled processes $X_{t}^{(\delta)}$, see (\ref{eqScaledProcess}), we have
\begin{equation}
 X^{(\delta)}_{t}(r)\leq X^{(0)}_{t}(r)\leq X^{(\delta)}_{t}(r)+\delta{t}^{-2/3}Z_{{t}+2{t}^{2/3}r}({t}).
\end{equation}
For any $\e>0$ it holds
\begin{equation}\begin{aligned}
  \Pb &\bigg(\bigcap_{k=1}^m\{X^{(\delta)}_{t}(r_k)\leq s_k\}\bigg)\geq\Pb\bigg(\bigcap_{k=1}^m\{X^{(0)}_{t}(r_k)\leq s_k\}\bigg)\\
  &\geq\Pb\bigg(\bigcap_{k=1}^m\{X^{(\delta)}_{t}(r_k)+\delta{t}^{-2/3}Z_{{t}+2{t}^{2/3}r}({t})\leq s_k\}\bigg)\\
  &\geq\Pb\bigg(\bigcap_{k=1}^m\{X^{(\delta)}_{t}(r_k)\leq s_k-\e\}\bigg)-\sum_{k=1}^m\Pb\left(\delta{t}^{-2/3}Z_{{t}+2{t}^{2/3}r}({t})>\e\right).
\end{aligned}\end{equation}
Then, taking ${t}\to\infty$, we obtain
\begin{equation}\begin{aligned}
  \Pb &\bigg(\bigcap_{k=1}^m\{\mathcal{A}_{\rm stat}^{(\delta)}(r_k)\leq s_k\}\bigg)
  \geq\limsup_{{t}\to\infty}\Pb\bigg(\bigcap_{k=1}^m\{X^{(0)}_{t}(r_k)\leq s_k\}\bigg)\\
  &\geq\liminf_{{t}\to\infty}\Pb\bigg(\bigcap_{k=1}^m\{X^{(0)}_{t}(r_k)\leq s_k\}\bigg)\\
  &\geq\Pb\bigg(\bigcap_{k=1}^m\{\mathcal{A}_{\rm stat}^{(\delta)}(r_k)\leq s_k-\e\}\bigg) -\sum_{k=1}^m\limsup_{{t}\to\infty}\Pb\left(Z_{{t}+2{t}^{2/3}r}({t})>{t}^{2/3}\e/\delta\right).
\end{aligned}\end{equation}
Using Proposition~\ref{propExPoint} on the last term and Proposition~\ref{propAnalyt} on the other terms, we can now take the limit $\delta\to0$, resulting in
\begin{equation}\begin{aligned}
  \Pb &\bigg(\bigcap_{k=1}^m\{\mathcal{A}_{\rm stat}(r_k)\leq s_k\}\bigg)
  \geq\limsup_{{t}\to\infty}\Pb\bigg(\bigcap_{k=1}^m\{X^{(0)}_{t}(r_k)\leq s_k\}\bigg)\\
  &\geq\liminf_{{t}\to\infty}\Pb\bigg(\bigcap_{k=1}^m\{X^{(0)}_{t}(r_k)\leq s_k\}\bigg)
  \geq\Pb\bigg(\bigcap_{k=1}^m\{\mathcal{A}_{\rm stat}(r_k)\leq s_k-\e\}\bigg).
\end{aligned}\end{equation}
Continuity of \eqref{eqAiryDef} in the $s_k$ finishes the proof.
\end{proof}

\begin{prop}\label{propExPoint}
For any $r\in\R$,
 \begin{equation}\label{eq5.14}
  \lim_{\beta\to\infty}\limsup_{{t}\to\infty}\Pb\left(Z_{{t}+2{t}^{2/3}r}({t})>\beta{t}^{2/3}\right)=0.
 \end{equation}
\end{prop}
\begin{proof}
 By scaling of ${t}$ and $\beta$, \eqref{eq5.14} is equivalent to
  \begin{equation}\label{eq5.15}
  \lim_{\beta\to\infty}\limsup_{{t}\to\infty}\Pb\left(Z_{{t}}({t}+2{t}^{2/3}r)>\beta{t}^{2/3}\right)=0,
 \end{equation}
for any $r\in\R$, which is the limit we are showing. We introduce some new events:
\begin{equation}
 \begin{aligned}
  M_\beta&:=\{Z_{{t}}({t}+2{t}^{2/3}r)>\beta{t}^{2/3}\}\\
  E_\beta&:=\{L_{(0,0)\to(\beta{t}^{2/3},0)}+L_{(\beta{t}^{2/3},0)\to({t}+2{t}^{2/3}r,{t})}\leq2{t}+2{t}^{2/3}r+s{t}^{1/3}\}\\
  N_\beta&:=\{L_{(0,0)\to({t}+2{t}^{2/3}r,{t})}\leq2{t}+2{t}^{2/3}r+{t}^{1/3}s\}.
 \end{aligned}
\end{equation}
Notice that if $M_\beta$ occurs, then
\begin{equation}
L_{(0,0)\to({t}+2{t}^{2/3}r,{t})}=L_{(0,0)\to(\beta{t}^{2/3},0)}+L_{(\beta{t}^{2/3},0)\to({t}+2{t}^{2/3}r,{t})},
\end{equation} resulting in $M_\beta\cap E_\beta\subseteq N_\beta$. We arrive at the inequality:
\begin{equation}\label{eqMbeta}
 \Pb(M_\beta)=\Pb(M_\beta\cap E_\beta)+\Pb(M_\beta\cap E_\beta^c)\leq\Pb(N_\beta)+\Pb(E_\beta^c).
\end{equation}
We further define new random variables
\begin{equation}\label{eq5.18}\begin{aligned}
 \xi_{\rm spiked}^{({t})}& =\frac{L_{(\beta{t}^{2/3},0)\to({t}+2{t}^{2/3}r,{t})}-2{t}-2{t}^{2/3}(r-\beta/2)}{{t}^{1/3}}+(r-\beta/2)^2,\\
 \xi_{\rm GUE}^{({t})}&=\frac{L^{\rm step}_{(0,1)\to({t}+2{t}^{2/3}r,{t})}-2{t}-2{t}^{2/3}r}{{t}^{1/3}}+r^2,\\
 \xi_{\rm N}^{({t})}&=\frac{L_{(0,0)\to(\beta{t}^{2/3},0)}-\beta{t}^{2/3}}{\sqrt{\beta}{t}^{1/3}}.
\end{aligned}\end{equation}
By Theorem 7 \cite{Wei11}, for any fixed $r\in\R$,
\begin{equation}\label{eqGUE}
 \xi_{\rm GUE}^{({t})}\stackrel{d}{\to} \xi_{\rm GUE},
\end{equation}
where $\xi_{\rm GUE}$ has the GUE Tracy-Widom distribution. $\xi_{\rm spiked}^{({t})}$ follows the distribution of the largest eigenvalue of a critically 	spiked GUE matrix, as will be shown in Lemma~\ref{lemSpike}. $\xi_{\rm N}^{({t})}$ has the distribution of a standard normal random variable $\xi_{\rm N}$ for any $\beta>0$, ${t}>0$.

Combining these definitions, we have:
\begin{equation}
 \begin{aligned}
  \Pb(E_\beta)=\Pb\big(\sqrt{\beta}\xi_{\rm N}^{({t})}+\xi_{\rm spiked}^{({t})}\leq  (r-\beta/2)^2+s\big).
 \end{aligned}
\end{equation}
Fix $s=3r^2-\beta^2/16$, such that:
\begin{equation}\begin{aligned}
  \left(r-\frac{\beta}{2}\right)^2+s&=4r^2-r\beta+\frac{\beta^2}{16}+\frac{\beta^2}{8}\geq \frac{\beta^2}{8}.
 \end{aligned}\end{equation}
Using the independence of $\xi_{\rm N}^{({t})}$ and $\xi_{\rm spiked}^{({t})}$, we obtain
\begin{equation}\label{eqEbeta}\begin{aligned}
  \Pb(E_\beta)&\geq\Pb\left(\sqrt{\beta}\xi_{\rm N}^{({t})}+\xi_{\rm spiked}^{({t})}\leq  \frac{\beta^2}{16}+\frac{\beta^2}{16}\right)\geq\Pb\left(\xi_{\rm N}^{({t})}\leq  \frac{\beta^{3/2}}{16}\text{ and }\xi_{\rm spiked}^{({t})}\leq\frac{\beta^2}{16}\right)\\
  &=\Pb\left(\xi_{\rm N}^{({t})}\leq  \frac{\beta^{3/2}}{16}\right)\Pb\left(\xi_{\rm spiked}^{({t})}\leq\frac{\beta^2}{16}\right)
 \end{aligned}\end{equation}
Further, the inequality
\begin{equation}
 L^{\rm step}_{(0,1)\to({t}+2{t}^{2/3}r,{t})}\leq L_{(0,0)\to({t}+2{t}^{2/3}r,{t})}
\end{equation}
leads to
\begin{equation}\label{eqNbeta}
 \Pb(N_\beta)\leq\Pb\big(\xi_{\rm GUE}^{({t})}\leq 4r^2-\beta^2/16\big).
\end{equation}
Inserting \eqref{eqEbeta} and \eqref{eqNbeta} into \eqref{eqMbeta}, we arrive at
\begin{equation}\begin{aligned}
 \Pb(M_\beta)\leq\Pb\left(\xi_{\rm GUE}^{({t})}\leq 4r^2-\frac{\beta^2}{16}\right)+1-\Pb\left(\xi_{\rm N}^{({t})}\leq  \frac{\beta^{3/2}}{16}\right)\Pb\left(\xi_{\rm spiked}^{({t})}\leq\frac{\beta^2}{16}\right)
\end{aligned}\end{equation}
By \eqref{eqGUE} and Lemma~\ref{lemSpike} we can take limits:
\begin{equation}\begin{aligned}
0&\leq\limsup_{\beta\to\infty}\limsup_{{t}\to\infty}\Pb\left(M_\beta\right)\\
&\leq\lim_{\beta\to\infty}\left[\Pb\left(\xi_{\rm GUE}\leq 4r^2-\frac{\beta^2}{16}\right)+1-\Pb\left(\xi_{\rm N}\leq  \frac{\beta^{3/2}}{16}\right)\Pb\left(\xi_{\rm spiked}(\beta)\leq\frac{\beta^2}{16}\right)\right]\\&=0.
\end{aligned}\end{equation}
\end{proof}

\begin{lem}\label{lemSpike}
Let $r\in\R$ be fixed. For any $\beta>2(r+1)$, as ${t}\to\infty$, the random variable
 \begin{equation}
   \xi_{\rm spiked}^{({t})}=\frac{L_{(\beta{t}^{2/3},0)\to({t}+2{t}^{2/3}r,{t})}-2{t}-2{t}^{2/3}(r-\beta/2)}{{t}^{1/3}}+(r-\beta/2)^2
 \end{equation}
converges in distribution,
\begin{equation}
 \xi_{\rm spiked}^{({t})}\stackrel{d}{\to} \xi_{\rm spiked}(\beta).
\end{equation}
In addition, $\xi_{\rm spiked}(\beta)$ satisfies
\begin{equation}
 \lim_{\beta\to\infty}\Pb\left(\xi_{\rm spiked}(\beta)\leq\beta^2/16\right)=1.
\end{equation}

\end{lem}
\begin{proof}
 The family of processes $L_{(\beta{t}^{2/3},0)\to(\beta{t}^{2/3}+t,n)}$ indexed by $n\in\N_0$ and time parameter $t\geq0$ is precisely a marginal of Warren's process with drifts, starting at zero, as defined in \cite{FF13}. In our case only the first particle has a drift of $1$, and all the others zero. By Theorem 2 \cite{FF13}, the fixed time distribution of this process is given by the distribution of the largest eigenvalue of a spiked $n\times n$ GUE matrix, where the spikes are given by the drifts.

 Thus we can apply the results on spiked random matrices, more concretely we want to apply Theorem 1.1 \cite{BW10}, with the potential $V(x)=-x^2/2$. Since
 \begin{equation}
  L^*:=L_{(\beta{t}^{2/3},0)\to({t}+2{t}^{2/3}r,n)}
 \end{equation}
 represents a $n\times n$ GUE matrix diffusion $M(t)$ at time $t={t}+2{t}^{2/3}(r-\beta/2)$, it is distributed according to the density
 \begin{equation}
  p_n(M)=\frac{1}{Z_n}\exp\left(-\frac{\Tr(M-t{\mathrm I}_{11})^2}{2t}\right),
 \end{equation}
 where ${\mathrm I}_{11}$ is a $n\times n$ matrix with a one at entry $(1,1)$ and zeros elsewhere. In order to apply the theorem we need the density given in equation (1) \cite{BW10}, i.e., consider the scaled quantity $L^*/\sqrt{nt}$. The size of the first-order spike is then:
 \begin{equation}
  a= t/\sqrt{nt} = \sqrt{1+2{t}^{-1/3}(r-\beta/2)} = 1+(r-\beta/2){t}^{-1/3}+\Or({t}^{-2/3}).
 \end{equation}
We are thus in the neighbourhood of the critical value ${\mathbf a}_c=1$. For $\alpha\geq0$, let
\begin{equation}
 C_\alpha(\xi)=\int_{-\infty}^0e^{\alpha x}\Ai(x+\xi) \D x.
\end{equation}
With $F_0(s)$ being the cumulative distribution function of the GUE Tracy-Widom distribution, and $K_{0,0}(s_1,s_2)$ as in \eqref{eqKernelDef}, define:
\begin{equation}\label{eqF1}
 F_1(s;\alpha)=F_0(s)\Big(1-\left\langle(1-P_s K_{0,0}P_s)^{-1}C_\alpha,P_s\Ai\right\rangle\Big).
\end{equation}

Applying (28) \cite{BW10}, we have
 \begin{equation}
  n^{2/3}(L^*/\sqrt{nt}-2)\to \xi_{\rm spiked}(\beta),
 \end{equation}
with
\begin{equation}
 \Pb\left(\xi_{\rm spiked}(\beta)\leq\beta^2/16\right)=F_1(\beta^2/16,\alpha),
\end{equation}
where $\alpha=\beta/2-r$. Since in our case $\alpha>1$, we can estimate:
\begin{equation}
 \left|C_\alpha(\xi)\right|\leq\int_{-\infty}^0e^{\alpha x}e^{-x-\xi} \D x=e^{-\xi}\frac{1}{\alpha-1}.
\end{equation}
Combining this with the usual bounds on the Airy kernel and the Airy function, we see that as $\beta\to\infty$, the scalar product in \eqref{eqF1} converges to zero and we are left with the limit of $F_0$ which is one.

On the other hand,
 \begin{equation}\begin{aligned}
  n^{2/3}(L^*/\sqrt{nt}-2)\leq s\quad\Leftrightarrow\quad L^*\leq\sqrt{nt}(2+n^{-2/3}s),
 \end{aligned}\end{equation}
  and
  \begin{equation}
   \sqrt{nt}(2+n^{-2/3}s)=2{t}+2{t}^{2/3}(r-\beta/2)+{t}^{1/3}\left(s-(r-\beta/2)^2\right)+\Or(1),
  \end{equation}
from which the claim follows.
\end{proof}

\begin{prop}\label{propAnalyt}
 The function $\delta\mapsto\delta^{-1}\det(\Id-\widehat{\mathcal{P}}\widehat{K}^\delta_{r_1})$ can be extended analytically in the domain $\delta\in\R$. Its value at $\delta=0$ is given by
 \begin{equation}
  G_m(\vec{r},\vec{s})\det\left(\Id-\mathcal{P}K\right)_{L^2(\R)}.
 \end{equation}
\end{prop}

\begin{proof}
We use the identity $\det(\Id+A)\det(\Id+B)=\det(\Id+A+B+AB)$ and Lemma~\ref{lemInv}  to factorize
\begin{equation}\label{eqFac}\begin{aligned}
	\delta^{-1}\det\big(\Id-\widehat{\mathcal{P}}\widehat{K}^\delta_{r_1}\big)&=\delta^{-1}\det(\Id-\widehat{\mathcal{P}}\widehat{K}^\delta_{r_1})
=\delta^{-1}\det\big(\Id-\widehat{\mathcal{P}}\widehat{K}-\delta\widehat{\mathcal{P}}\widehat{f}_{r_1}\otimes \widehat{g}_{r_1}\big)\\
	&=\delta^{-1}\det\big(\Id-\delta(\Id-\widehat{\mathcal{P}}\widehat{K})^{-1}\widehat{\mathcal{P}}\widehat{f}_{r_1}\otimes \widehat{g}_{r_1}\big)\cdot\det\big(\Id-\widehat{\mathcal{P}}\widehat{K}\big)\\
	&=\big(\delta^{-1}-\big\langle (\Id-\widehat{\mathcal{P}}\widehat{K})^{-1}\widehat{\mathcal{P}}\widehat{f}_{r_1},\widehat{g}_{r_1}\big\rangle\big)
\cdot\det\big(\Id-\widehat{\mathcal{P}}\widehat{K}\big)_{L^2(\R)}\\
	&=\big(\delta^{-1}-\big\langle (\Id-\mathcal{P}K)^{-1}\mathcal{P}f_{r_1},g_{r_1}\big\rangle\big)\cdot\det\big(\Id-\mathcal{P}K\big)_{L^2(\R)}.
\end{aligned}\end{equation}
Since the second factor is independent of $\delta$, the remaining task is the analytic continuation of the first. Using \eqref{contInt}, decompose $f_{r_1}$  as
\begin{equation}\label{eqfDec}
	f_{r_1}(s)=1+\frac{1}{2\pi\I}\int_{\rangle0} dW\frac{e^{-W^3/3-r_1W^2+sW}}{W}=:1+f^*(s).
\end{equation}
Now,
\begin{equation}\label{eq49}\begin{aligned}
		\langle P_{s_1}\mathbf{1},g_{r_1}\rangle&=\int_{s_1}^\infty ds\,\frac{1}{2\pi\I}\int_{0\langle\delta} dZ\frac{e^{Z^3/3+r_1Z^2-sZ}}{Z-\delta}\\
&=\frac{1}{2\pi\I}\int_{0\langle\delta} dZ\frac{e^{Z^3/3+r_1Z^2-s_1Z}}{Z(Z-\delta)}\\&=\frac{1}{\delta}+\frac{1}{2\pi\I}\int_{\langle0,\delta} dZ\frac{e^{Z^3/3+r_1Z^2-s_1Z}}{Z(Z-\delta)}=:\frac{1}{\delta}-\mathcal{R}_\delta.
\end{aligned}\end{equation}
The function $\mathcal{R}_\delta$ is analytic in $\delta\in\R$. Using these two identities as well as $(\Id-\mathcal{P}K)^{-1}=\Id+(\Id-\mathcal{P}K)^{-1}\mathcal{P}K$, we can rearrange the inner product as follows:
\begin{equation}\label{eq51}\begin{aligned}
	\frac{1}{\delta}-&\left\langle (\Id-\mathcal{P}K)^{-1}\mathcal{P}f_{r_1},g_{r_1}\right\rangle\\&=\frac{1}{\delta}-\left\langle(\Id-\mathcal{P}K)^{-1}\mathcal{P}\mathbf{1}+(\Id-\mathcal{P}K)^{-1}\mathcal{P}f^*,g_{r_1}\right\rangle
	\\&=\frac{1}{\delta}-\left\langle\mathcal{P}\mathbf{1}+(\Id-\mathcal{P}K)^{-1}(\mathcal{P}K\mathcal{P}\mathbf{1}+\mathcal{P}f^*),g_{r_1}\right\rangle
\\&=\frac{1}{\delta}-\left\langle P_{s_1}\mathbf{1},g_{r_1}\right\rangle-\left\langle(\mathcal{P}-P_{s_1})\mathbf{1}+(\Id-\mathcal{P}K)^{-1}(\mathcal{P}K\mathcal{P}\mathbf{1}+\mathcal{P}f^*),g_{r_1}\right\rangle\\
&=\mathcal{R}_\delta-\left\langle(\Id-\mathcal{P}K)^{-1}\left(\mathcal{P}f^*+\mathcal{P}KP_{s_1}\mathbf{1}+(\mathcal{P}-P_{s_1})\mathbf{1}\right),g_{r_1}\right\rangle
\end{aligned}\end{equation}
Since $g_{r_1}$ is evidently analytic in $\delta\in\R$, we are left to show convergence of the scalar product.

All involved functions are locally bounded, so to establish convergence it is enough to investigate their asymptotic behaviour. $g_{r_1}$ may grow exponentially at arbitrary high rate, depending on $r_1$ and $\delta$, for both large positive and large negative arguments. We therefore need superexponential bounds on the function:
\begin{equation}
 (\Id-\mathcal{P}K)^{-1}\left(\mathcal{P}f^*+\mathcal{P}KP_{s_1}\mathbf{1}+(\mathcal{P}-P_{s_1})\mathbf{1}\right).
\end{equation}
For this purpose we first need an expansion of the operator $\mathcal{P}$:
\begin{equation}\label{eqPExp}
 \mathcal{P}=\sum_{k=1}^n\bar{P}_{s_1}V_{r_1,r_2}\dots \bar{P}_{s_{k-1}}V_{r_{k-1},r_k}P_{s_k}V_{r_k,r_1}.\\
\end{equation}
Notice that all operators $P_{s_i}$, $\bar{P}_{s_i}$ and $V_{r_i,r_j}$ map superexponentially decaying functions onto superexponentially decaying functions. Moreover $P_{s_i}$ and $\bar{P}_{s_i}$ generate superexponential decay for large negative resp. positive arguments.

The function $f^*$ decays superexponentially for large arguments but may grow exponentially for small ones. Since every part of the sum contains one projection $P_{s_k}$, $\mathcal{P}f^*$ decays superexponentially on both sides.

Examining $(\mathcal{P}-P_{s_1})\mathbf{1}$, notice that the $k=1$ contribution in \eqref{eqPExp} is equal to $P_{s_1}$, which is cancelled out here. All other contributions contain both $\bar{P}_{s_1}$ and $P_{s_k}$, which ensure superexponential decay.

Using the usual asymptotic bound on the Airy function, we see that the operator $K$ maps any function in its domain onto one which is decreasing superexponentially for large arguments. By previous arguments, functions in the image of $\mathcal{P}K$ decay on both sides, in particular $\mathcal{P}KP_{s_1}\mathbf{1}$.

Now, in order to establish the finiteness of the scalar product, decompose the inverse operator as $(\Id-\mathcal{P}K)^{-1}=\Id+\mathcal{P}K(\Id-\mathcal{P}K)^{-1}$. The contribution coming from the identity has just been settled. As inverse of a bounded operator, $(\Id-\mathcal{P}K)^{-1}$ is also bounded. Because of the rapid decay, the functions $\mathcal{P}f^*$, $\mathcal{P}KP_{s_1}\mathbf{1}$ and $(\mathcal{P}-P_{s_1})\mathbf{1}$ are certainly in $L^2(\R)$ and thus mapped onto $L^2(\R)$ by this operator. Finally, the image of an $L^2(\R)$-function under the operator $\mathcal{P}K$ is decaying superexponentially on both sides.

The expression \eqref{eq51} is thus an analytic function in $\delta$ in the domain $\R$. Setting $\delta=0$ returns the value of $G_m(\vec{r},\vec{s})$. Combining these results with \eqref{eqFac} finishes the proposition.
\end{proof}
\begin{lem}\label{lemInv}
 The operator $\Id-\mathcal{P}K$ is invertible.
\end{lem}
\begin{proof}
 We employ the same strategy as in~\cite{BFP09}. For that purpose we use the following equivalence
 \begin{equation}
  \det(\Id+A)\neq 0\ \Longleftrightarrow\ \Id+A \text{ is invertible}.
 \end{equation}
Let $s_{\rm min}=\min_ks_k$.
 \begin{equation}\begin{aligned}
  \det&(\Id-\mathcal{P}K)=\Pb\bigg(\bigcap_{k=1}^m\{\mathcal{A}_2(r_k)-r_k^2\leq s_k\}\bigg)\\
  &\geq \Pb\bigg(\bigcap_{k=1}^m\{\mathcal{A}_2(r_k)-r_k^2\leq s_{\rm min}\}\bigg)\geq \Pb\Big(\max_{r\in\R}(\mathcal{A}_2(r)-r^2)\leq s_{\rm min}\Big)\\
  &=F_{\rm GOE}(2^{2/3}s_{\rm min})>0
  \end{aligned}\end{equation}
for any $s_{\rm min}>-\infty$, where $F_{\rm GOE}$ is the GOE Tracy-Widom distribution function. For the last equality see~\cite{Jo03b,CQR11}. The tails of the GOE Tracy-Widom distribution have been studied in great detail in various publications, see for instance~\cite{BBdF08}.
\end{proof}

\section{Gaussian increments}\label{SectGaussianIncr}
In this section we prove that the Airy$_{\rm stat}$ process has Brownian increments for nonnegative arguments:
\begin{thm}
 Let $0\leq r_1<r_2<\dots<r_m$. Then
 \begin{equation}
  \Pb\bigg(\bigcap_{k=2}^{m}\{\mathcal{A}_{\rm stat}(r_{k})-\mathcal{A}_{\rm stat}(r_{k-1})\in\D\sigma_k\}\bigg) =\prod_{k=2}^{m}\frac{e^{-\sigma_k^2/4(r_k-r_{k-1})}}{\sqrt{4\pi(r_k-r_{k-1})}}\,\D\vec\sigma.
 \end{equation}
\end{thm}

\begin{proof}
 Without loss of generality we can assume that $r_1=0$. Denoting the partial derivative with respect to the $i$-th coordinate by $\partial_i$, we have
 \begin{equation}
	\Pb\bigg(\bigcap_{k=1}^m\{\mathcal{A}_{\rm stat}(r_k)\leq s_k\}\bigg) = \sum_{i=1}^m\partial_i\Lambda\left(s_1,\dots,s_m\right),
\end{equation}
with
 \begin{equation}
	\Lambda\left(s_1,\dots,s_m\right)=G_m(\vec{r},\vec{s})\det\left(\Id-\mathcal{P}K\right)_{L^2(\R)}.
\end{equation}
With a small abuse of notations, in what follows we will write
\begin{equation}
\Pb\bigg(\bigcap_{k=1}^m\{\mathcal{A}_{\rm stat}(r_k)\in\D s_k\}\bigg) \equiv \Pb\bigg(\bigcap_{k=1}^m\{\mathcal{A}_{\rm stat}(r_k)=s_k\}\bigg) \D s_1\cdots \D s_m.
\end{equation}
Then,
\begin{equation}
\Pb\bigg(\bigcap_{k=1}^m\{\mathcal{A}_{\rm stat}(r_k)=s_k\}\bigg) = \prod_{i=1}^m\partial_i\sum_{j=1}^m\partial_j\Lambda\left(s_1,\dots,s_m\right).
\end{equation}
The crucial identity is:
 \begin{equation}\label{eq8.4}\begin{aligned}
	\Pb&\bigg(\bigcap_{k=2}^{m}\{\mathcal{A}_{\rm stat}(r_{k})-\mathcal{A}_{\rm stat}(r_{k-1})= \sigma_k\}\bigg) \\
	&=\int_\R\D\sigma_1  \Pb\bigg(\bigcap_{k=1}^m\{\mathcal{A}_{\rm stat}(r_k)= \sigma_1+\dots+\sigma_k\}\bigg)\\ &=\int_\R\D\sigma_1\bigg(\prod_{i=1}^m\partial_i\sum_{j=1}^m\partial_j\bigg)\Lambda\left(\sigma_1,\sigma_1+\sigma_2,\dots,\sigma_1+\dots+\sigma_m\right)\\ &=\int_\R\D\sigma_1\frac{\D}{\D\sigma_1}\bigg(\prod_{i=1}^m\partial_i\bigg)\Lambda\left(\sigma_1,\sigma_1+\sigma_2,\dots,\sigma_1+\dots+\sigma_m\right)\\
&=\bigg(\prod_{i=1}^m\partial_i\bigg)\Lambda\left(\sigma_1,\sigma_1+\sigma_2,\dots,\sigma_1+\dots+\sigma_m\right)\bigg|_{\sigma_1=-\infty}^{\sigma_1=\infty}.
\end{aligned}\end{equation}
We therefore have to study the asymptotics of $\Lambda$ as $\sigma_1\to\pm\infty$.

First we decompose $\Lambda$ as
\begin{equation}\begin{aligned}
 \Lambda&=\Lambda_1+\Lambda_2,\\
 \Lambda_1&:=\left(\mathcal{R}-1\right)\det\left(\Id-\mathcal{P}K\right)_{L^2(\R)},\\
 \Lambda_2&:=\det\Big(\Id-\mathcal{P}K-\left(\mathcal{P}f^*+\mathcal{P}KP_{s_1}\mathbf{1}+(\mathcal{P}-P_{s_1})\mathbf{1}\right)\otimes g\Big)_{L^2(\R)}.
\end{aligned}\end{equation}
Since $r_1=0$ some functions simplify as
\begin{equation}\begin{aligned}
	\mathcal{R}&=s_1+\int_{s_1}^\infty \D x\int_x^\infty \D y\, \Ai(y),\\
	f^*(s)&=-\int_s^\infty \D x\, \Ai(x),\\
	g(s)&=1-\int_s^\infty \D x\, \Ai(x)=\int_{-\infty}^s \D x\, \Ai(x),\\
	K(s_1,s_2)&=\int_0^\infty\D x\, \Ai(s_1+x)\Ai(s_2+x),
\end{aligned}\end{equation}
where we used the identity (D.2) from \cite{FS05a}.

Now consider $\Lambda_1$.
\begin{equation}
 \bigg(\prod_{i=1}^m\partial_i\bigg)\Lambda_1(\vec{s})=(\mathcal{R}-1)\bigg(\prod_{i=1}^m\partial_i\bigg)
 \det(\Id-\mathcal{P}K)+\partial_1\mathcal{R}\bigg(\prod_{i=2}^m\partial_i\bigg)\det(\Id-\mathcal{P}K).
\end{equation}
Regarding the first term, notice that the multiple derivative of the Fredholm determinant gives exactly the multipoint density of the Airy$_2$ process, which is known to decay exponentially for both large positive and negative arguments. This exponential decay dominates over the linear growth of $\mathcal{R}$. Similarly, the $(m-1)$-fold derivative is smaller the $(m-1)$-point density of the Airy$_2$ process, so this contribution vanishes in the limit, too.

Continuing to $\Lambda_2$, using $f^*=-K\mathbf{1}$, we first simplify the expression
\begin{equation}
 \Lambda_2=\det\Big(\Id-\mathcal{P}K+\left(\mathcal{P}K\bar{P}_{s_1}\mathbf{1}-(\mathcal{P}-P_{s_1})\mathbf{1}\right)\otimes g\Big)_{L^2(\R)}
\end{equation}
We introduce the shift operator $S$, $(Sf)(x)=f(x+\sigma_1)$, which satisfies $SV_{r_i,r_j}S^{-1}=V_{r_i,r_j}$ and $P_{a+\sigma_1}=S^{-1}P_aS$, and consequently also
\begin{equation}
 \Id-\bar{P}_{s_1+\sigma_1}V_{r_1,r_2}\bar{P}_{s_2+\sigma_1}\cdots V_{r_{m-1},r_m}\bar{P}_{s_m+\sigma_1}V_{r_m,r_1}=S^{-1}\mathcal{P}S.
\end{equation}
Using $\det(\Id-AB)=\det(\Id-BA)$, we have
\begin{equation}
 \Lambda_2(\vec{s}+\sigma_1)=\det\Big(\Id-\mathcal{P}SKS^{-1}+\left(\mathcal{P}SKS^{-1}\bar{P}_{s_1}\mathbf{1}-(\mathcal{P}-P_{s_1})\mathbf{1}\right)\otimes Sg\Big)_{L^2(\R)}.
\end{equation}
Now the dependence on the vector $\vec{s}$ is only in the projection operators, while the dependence on $\sigma_1$ is only in these two operators:
\begin{equation}\begin{aligned}
	(Sg)(s)&=\int_{-\infty}^{s+\sigma_1} \D x\, \Ai(x),\\
	(SKS^{-1})(s_1,s_2)&=\int_{\sigma_1}^\infty\D x\, \Ai(s_1+x)\Ai(s_2+x).
\end{aligned}\end{equation}
For large $\sigma_1$, we have $Sg\to\mathbf{1}$ and $SKS^{-1}\to0$ (both strong types of convergence from the superexponential Airy decay). So
\begin{equation}
 \lim_{\sigma_1\to\infty}\Lambda_2(\vec{s}+\sigma_1)=\det\Big(\Id-(\mathcal{P}-P_{s_1})\mathbf{1}\otimes \mathbf{1}\Big)_{L^2(\R)}=1-\langle(\mathcal{P}-P_{s_1})\mathbf{1}, \mathbf{1}\rangle_{L^2(\R)}
\end{equation}
Applying the expansion \eqref{eqPExp}, we arrive at:
\begin{equation}\label{eq8.13}\begin{aligned}
 \bigg(\prod_{i=1}^m\partial_i\bigg)\lim_{\sigma_1\to\infty}\Lambda_2(\vec{s}+\sigma_1)&=
 -\bigg(\prod_{i=1}^m\partial_i\bigg)\sum_{k=2}^m\langle\bar{P}_{s_1}V_{r_1,r_2}\dots \bar{P}_{s_{k-1}}V_{r_{k-1},r_k}P_{s_k}\mathbf{1}, \mathbf{1}\rangle\\
 &=-\bigg(\prod_{i=1}^m\partial_i\bigg)\langle\bar{P}_{s_1}V_{r_1,r_2}\dots \bar{P}_{s_{m-1}}V_{r_{m-1},r_m}P_{s_m}\mathbf{1}, \mathbf{1}\rangle
\end{aligned}\end{equation}
Writing out this scalar product and applying the fundamental theorem of calculus leads to:
\begin{equation}
 \eqref{eq8.13}=V_{r_1,r_2}(s_1,s_2)V_{r_2,r_3}(s_2,s_3)\dots V_{r_{m-1},r_m}(s_{m-1},s_m),
\end{equation}
which is the desired Gaussian density after setting $s_i=\sum_{k=2}^i\sigma_k$ as in \eqref{eq8.4}.

For large negative $\sigma_1$, we have $Sg\to\mathbf{0}$ and $SKS^{-1}\to\Id$. The rank one contribution is thus
\begin{equation}
 \left(\mathcal{P}\bar{P}_{s_1}\mathbf{1}-(\mathcal{P}-P_{s_1})\mathbf{1}\right)\otimes 0.
\end{equation}
We have to be somewhat careful here, as the convergence is weak (only pointwise) and $(Sg)$ is not even $L^2$-integrable. But the first factor decays superexponentially on both sides for finite $\sigma_1$ and also in the limiting case $\mathcal{P}\bar{P}_{s_1}\mathbf{1}-(\mathcal{P}-P_{s_1})\mathbf{1}=(1-\mathcal{P})P_{s_1}\mathbf{1}$, so one should be able to derive nice convergence properties. Neglecting this rank one contribution we are left with
\begin{equation}
 \lim_{\sigma_1\to-\infty}\Lambda_2(\vec{s}+\sigma_1)=\det\Big(\Id-\mathcal{P}\Id\Big)_{L^2(\R)}=0.
\end{equation}
\end{proof}


\end{document}